\documentclass[journal,twocolumn]{IEEEtran}
\IEEEoverridecommandlockouts
\usepackage{amsthm}
\usepackage{arydshln}
\usepackage{multirow}
\usepackage{calc}
\usepackage{blkarray}
\usepackage{url}
\theoremstyle{definition}
\newtheorem{theorem}{Theorem}
\newtheorem{lemma}{Lemma}
\newtheorem{claim}{Claim}

\newtheorem{corollary}{Corollary}
\newtheorem{example}{Example}
\newtheorem{remark}{Remark}
\newtheorem{definition}{Definition}
\usepackage{graphicx,cite}
\usepackage{dblfloatfix}

\usepackage{blindtext, graphicx, amsfonts,amssymb,multirow,epstopdf}
\usepackage{dsfont}

\ifCLASSINFOpdf

\else

\fi

\usepackage[cmex10]{amsmath}
\usepackage[linesnumbered,ruled]{algorithm2e}

\SetCommentSty{mycommfont}
\newcommand{\calA}{\mathcal{A}}
\newcommand{\calC}{\mathcal{C}}

\newcommand{\calN}{\mathcal{N}}
\newcommand{\calP}{\mathcal{P}}
\newcommand{\calS}{\mathcal{S}}

\newcommand{\calT}{\mathcal{T}}
\newcommand{\calI}{\mathcal{I}}

\newcommand{\bfg}{\mathbf{g}}

\newcommand{\bfv}{\mathbf{v}}

\newcommand{\bfu}{\mathbf{u}}
\newcommand{\bfc}{\mathbf{c}}
\newcommand{\bfT}{\mathbf{T}}
\newcommand{\bfG}{\mathbf{G}}
\newcommand{\bfS}{\mathbf{S}}
\newcommand{\bfI}{\mathbf{I}}

\newcommand{\bfC}{\mathbf{C}}
\newcommand{\bfD}{\mathbf{D}}
\newcommand{\bfA}{\mathbf{A}}
\newcommand{\bfB}{\mathbf{B}}

\newcommand{\bfF}{\mathbf{F}}
\newcommand{\bfzr}{\mathbf{0}}
\newcommand{\bfoe}{\mathbf{1}}

\begin{document}

	%
	\title{Coded Caching Schemes with Reduced Subpacketization from Linear Block Codes}
\author{\IEEEauthorblockN{Li Tang and Aditya Ramamoorthy\\}
		\IEEEauthorblockA{Department of Electrical and Computer Engineering\\
		Iowa State University\\
		Ames, IA 50010\\
			Emails:\{litang, adityar\}@iastate.edu}
	\thanks{This work was supported in part by the National Science Foundation by grants CCF-1718470, CCF-1320416 and CCF-1149860. This paper was presented in part at the 2016 IEEE Workshop on Network Coding and Applications (NetCod) and at the 2017 IEEE International Symposium on Information Theory (ISIT).}}

	\maketitle
\begin{abstract}
Coded caching is a technique that generalizes conventional caching and promises significant reductions in traffic over caching networks.
However, the basic coded caching scheme requires that
each file hosted in the server be partitioned into a large number
(i.e., the subpacketization level) of non-overlapping subfiles. From
a practical perspective, this is problematic as it means that prior
schemes are only applicable when the size of the files is extremely
large. In this work, we propose coded caching schemes based on combinatorial structures called resolvable designs.
These structures can be obtained in a natural manner from linear block codes whose generator matrices possess certain rank properties.
We obtain several schemes with subpacketization levels substantially lower than the basic scheme at the cost of an increased rate. Depending on the system parameters, our approach allows
us to operate at various points on the subpacketization level vs. rate tradeoff.
\end{abstract}
\section{Introduction}
\label{sec:intro}

Caching is a popular technique for facilitating large scale content delivery over the Internet. Traditionally, caching operates by storing popular content closer to the end users. Typically, the cache serves an end user's file request partially (or sometimes entirely) with the remainder of the content coming from the main server. Prior work in this area \cite{maddahN14} demonstrates that allowing coding in the cache and coded transmission from the server (referred to as {\it coded caching}) to the end users can allow for significant reductions in the number of bits transmitted from the server to the end users. This is an exciting development given the central role of caching in supporting a significant fraction of Internet traffic. In particular, reference \cite{maddahN14} considers a scenario where a single server contains $N$ files. 
The server connects to $K$ users over a shared link and each user has a cache that allows it to store $M/N$ fraction of all the files in the server. Coded caching consists of two distinct phases: a $\emph{placement phase}$ and a $\emph{delivery phase}$. In the placement phase, the caches of the users are populated. This phase does not depend on the user demands which are assumed to be arbitrary. In the delivery phase, the server sends a set of $\emph{coded}$ signals that are broadcast to each user such that each user's demand is satisfied.

The original work of \cite{maddahN14} considered the case of centralized coded caching, where the server decides the content that needs to be placed in the caches of the different users. Subsequent work considered the decentralized case where the users populate their caches by randomly choosing parts of each file while respecting the cache size constraint. Recently, there have been several papers that have examined various facets of coded caching. These include tightening known bounds on the coded caching rate \cite{ghasemiR17_accepted,sengupta2015improved}, considering
 issues with respect to decentralized caching \cite{maddahN14mr_tradeoff}, explicitly considering popularities of files \cite{maddahN14nonuniform_demand,hachemKD14a}, network topology issues \cite{TangR16,JJ15} and synchronization issues \cite{ghasemiR17,niesen2015coded}. 

In this work, we examine another important aspect of the coded caching problem that is closely tied to its adoption in practice. It is important to note that the huge gains of coded caching require each file to be partitioned into $F_s \approx \binom{K}{\frac{KM}{N}}$ non-overlapping subfiles of equal size; $F_s$ is referred to as the \emph{subpacketization level}. It can be observed that for a fixed cache size $\frac{M}{N}$, $F_s$ grows exponentially with $K$. This can be problematic in practical implementations. For instance, suppose that $K=64$, with $\frac{M}{N}=0.25$ so that $F_s=\binom{64}{16}\approx 4.8 \times 10^{14}$ with a rate $R \approx 2.82$. In this case, it is evident that at the bare minimum, the size of each file has to be at least $480$ terabits for leveraging the gains in \cite{maddahN14}. It is even worse in practice. The atomic unit of storage on present day hard drives is a sector of size $512$ bytes and the trend in the disk drive industry is to move this to $4096$ bytes \cite{toshiba_wp}. As a result, the minimum size of each file needs to be much higher than $480$ terabits. Therefore, the scheme in \cite{maddahN14} is not practical even for moderate values of $K$. Furthermore, even for smaller values of $K$, schemes with low subpacketization levels are desirable. This is because any practical scheme will require each of the subfiles to have some header information that allows for decoding at the end users. When there are a large number of subfiles, the header overhead may be non-negligible. For these same parameters ($K=64, M/N = 0.25$) our proposed approach in this work allows us obtain, e.g., the following operating points: (i) $F_s \approx 1.07 \times 10^9$ and $R = 3$, (ii) $F_s \approx 1.6\times 10^4$ and $R=6$, (iii) $F_s=64$ and $R=12$. For the first point, it is evident that the subpacketization level drops by over five orders of magnitude with only a very small increase in the rate. Point (ii) and (iii) show proposed scheme allows us to operate at various points on the tradeoff between subpacketization level and rate. 


The issue of subpacketization was first considered in the work of \cite{shanmugam_et_al14,shanmugam_et_al16} in the decentralized coded caching setting. In the centralized case it was considered in the work of \cite{yan_et_al17}. They proposed a low subpacketization scheme based on placement delivery arrays. Reference \cite{shangguan2016centralized} viewed the problem from a hypergraph perspective and presented several classes of coded caching schemes. The work of \cite{shanmugam2017coded} has recently shown that there exist coded caching schemes where the subpacketization level grows linearly with the number of users $K$; however, this result only applies when the number of users is very large. We elaborate on related work in Section \ref{sec:related_work}.

In this work, we propose low subpacketization level schemes for coded caching. Our proposed schemes leverage the properties of combinatorial structures known as resolvable designs and their natural relationship with linear block codes. Our schemes are applicable for a wide variety of parameter ranges and allow the system designer to tune the subpacketization level and the gain of the system with respect to an uncoded system. We note here that designs have also been used to obtain results in distributed data storage \cite{olmezR16} and network coding based function computation in recent work \cite{tripathyR15,tripathyR17}.


This paper is organized as follows. Section \ref{sec:background} discusses the background and related work and summarizes the main contributions of our work. Section \ref{sec:design} outlines our proposed scheme. It includes all the constructions and the essential proofs. A central object of study in our work are matrices that satisfy a property that we call the consecutive column property (CCP). Section \ref{sec:LinearCode} overviews several constructions of matrices that satisfy this property. Several of the longer and more involved proofs of statements in Sections \ref{sec:design} and \ref{sec:LinearCode} appear in the Appendix. In Section \ref{sec:compare} we perform an in-depth comparison our work with existing constructions in the literature. We conclude the paper with a discussion of opportunities for future work in Section \ref{sec:conclusion}.
\section{Background, Related Work and Summary of Contributions}
\label{sec:background}
We consider a scenario where the server has $N$ files each of which consist of $F_s$ subfiles. There are $K$ users each equipped with a cache of size $MF_s$ subfiles. The coded caching scheme is specified by means of the placement scheme and an appropriate delivery scheme for each possible demand pattern. In this work, we use combinatorial designs \cite{Stinson} to specify the placement scheme in the coded caching system.	
\begin{definition}
\label{def:comb_design}
A design is a pair $(X, \calA)$ such that
\begin{enumerate}
\item  $X$ is a set of elements called points, and
\item $\calA$ is a collection of nonempty subsets of $X$ called blocks, where each block contains the same number of points.
\end{enumerate}
\end{definition}
A design is in one-to-one correspondence with an incidence matrix $\calN$ which is defined as follows.
\begin{definition}
	The incidence matrix $\calN$ of a design $(X, \calA)$ is a binary matrix of dimension $|X| \times |\calA|$, where the rows and columns correspond to the points and blocks respectively.
	Let $i \in X$ and $j \in \calA$. Then,
	\begin{align*}
		\calN (i,j) = \begin{cases}
			1 & \text{~if $i \in j$,}\\
			0 & \text{~otherwise}.
		\end{cases}
	\end{align*}
\end{definition}
It can be observed that the transpose of an incidence matrix also specifies a design. We will refer to this as the transposed design. In this work, we will utilize resolvable designs which are a special class of designs.
\begin{definition}
\label{def:resolv_design}
A parallel class $\calP$ in a design $(X,\calA)$ is a subset of disjoint blocks from $\calA$ whose union is $X$. A partition of $\calA$ into several parallel classes is called a resolution, and $(X,\calA)$ is said to be a resolvable design if $\calA$ has at least one resolution.
\end{definition}
For resolvable designs, it follows that each point also appears in the same number of blocks.

\begin{example}
\label{eg:resolv_6_2}
Consider a block design specified as follows.
\begin{align*}
X&=\{1,2,3,4\}, \text{~and}\\
\calA &=\{\{1,2\},\{1,3\},\{1,4\},\{2,3\},\{2,4\},\{3,4\}\}.
\end{align*}
Its incidence matrix is given below.
\begin{align*}
 \calN &= \begin{bmatrix}
 	1 & 1 & 1 & 0 & 0 & 0\\
 	1 & 0 & 0 & 1 & 1 & 0\\
 	0 & 1 & 0 & 1 & 0 & 1\\
 	0 & 0 & 1 & 0 & 1 & 1
\end{bmatrix}.
\end{align*}
It can be observed that this design is resolvable with the following parallel classes.
\begin{align*}
\calP_1 &= \{\{1,2\}, \{3,4\}\},\\
\calP_2 &= \{\{1,3\}, \{2,4\}\}, \text{~and}\\
\calP_3 &= \{\{1,4\}, \{2,3\}\}.
\end{align*}
\end{example}
In the sequel we let $[n]$ denote the set $\{1, \dots, n\}$.
We emphasize here that the original scheme of \cite{maddahN14} can be viewed as an instance of the trivial design. For example, consider the setting when $t = KM/N$ is an integer. Let  $X = [K]$ and $\calA = \{B: B \subset [K], |B| = t\}$. In the scheme of \cite{maddahN14}, the users are associated with $X$ and the subfiles with $\calA$. User $i \in [K]$ caches subfile $W_{n,B}, n \in [N]$ for $B \in \calA$ if $i \in B$. 
The main message of our work is that carefully constructed resolvable designs can be used to obtain coded caching schemes with low subpacketization levels, while retaining much of the rate gains of coded caching. The basic idea is to associate the users with the blocks and the subfiles with the points of the design. The roles of the users and subfiles can also be interchanged by simply working with the transposed design.

\begin{example}
\label{eg:placement_q_2_k_3}
Consider the resolvable design from Example \ref{eg:resolv_6_2}. The blocks in $\calA$ correspond to six users $U_{12}$, $U_{34}$, $U_{13}$, $U_{24}$, $U_{14}$, $U_{23}$. Each file is partitioned into $F_s=4$ subfiles $W_{n,1}, W_{n,2}, W_{n,3}, W_{n,4}$ which correspond to the four points in $X$. The cache in user $U_{B}$, denoted $Z_{B}$ is specified as $Z_{ij}=(W_{n,i},W_{n,j})_{n=1}^N$. For example, $Z_{12}=(W_{n,1}, W_{n,2})_{n=1}^{N}$.

We note here that the caching scheme is symmetric with respect to the files in the server. Furthermore, each user caches half of each file so that $M/N = 1/2$.
Suppose that in the delivery phase user $U_B$ requests file $W_{d_{B}}$ where $d_B \in [N]$. These demands can be satisfied as follows. We pick three blocks, one each from parallel classes $\calP_1$, $\calP_2$, $\calP_3$ and generate the signals transmitted in the delivery phase as follows.
\begin{align}
	&W_{d_{12},3}\oplus W_{d_{13},2}\oplus W_{d_{23},1}, \label{eq:eg_2_intro}\\
    &W_{d_{12},4}\oplus W_{d_{24},1}\oplus W_{d_{14},2}, \nonumber\\
    &W_{d_{34},1}\oplus W_{d_{13},4}\oplus W_{d_{14},3}, \text{~and} \nonumber\\
    &W_{d_{34},2}\oplus W_{d_{24},3}\oplus W_{d_{23},4}. \nonumber
\end{align}
The three terms in the in eq. (\ref{eq:eg_2_intro}) above correspond to blocks from different parallel classes $\{1,2\}\in \calP_1, \{1,3\}\in \calP_2 ,\{2,3\}\in \calP_3$. This equation has the {\it all-but-one} structure that was also exploited in \cite{maddahN14}, i.e., eq. (\ref{eq:eg_2_intro}) is such that each user caches all but one of the subfiles participating in the equation. Specifically, user $U_{12}$ contains $W_{n,1}$ and $W_{n,2}$ for all $n \in [N]$. Thus, it can decode subfile $W_{d_{12},3}$ that it needs. A similar argument applies to users $U_{13}$ and $U_{23}$. It can be verified that the other three equations also have this property. Thus, at the end of the delivery phase, each user obtains its missing subfiles.


This scheme corresponds to a subpacketization level of $4$ and a rate of $1$. In contrast, the scheme of \cite{maddahN14} would require a subpacketization level of $\binom{6}{3} = 20$ with a rate of $0.75$. Thus, it is evident that we gain significantly in terms of the subpacketization while sacrificing some rate gains.
\end{example}
As shown in Example \ref{eg:placement_q_2_k_3}, we can obtain a scheme by associating the users with the blocks and the subfiles with the points. In this work, we demonstrate that this basic idea can be significantly generalized and several schemes with low subpacketization levels that continue to leverage much of the rate benefits of coded caching can be obtained. 

\subsection{Discussion of Related Work}
\label{sec:related_work}
Coded caching has been the subject of much investigation in recent work as discussed briefly earlier on. We now overview existing literature on the topic of low subpacketization schemes for coded caching.
In the original paper \cite{maddahN14}, for given problem parameters $K$ (number of users) and $M/N$ (cache fraction), the authors showed that when $N \geq K$, the rate equals
\begin{align*}
R &= \frac{K(1 - M/N)}{1 + KM/N}
\end{align*}
when $M$ is an integer multiple of $N/K$. Other points are obtained via memory sharing.
Thus, in the regime when $KM/N$ is large, the coded caching rate is approximately $N/M - 1$, which is independent of $K$. Crucially, though this requires the subpacketization level $F_s \approx \binom{K}{KM/N}$. It can be observed that for a fixed $M/N$, $F_s$ grows exponentially with $K$. This is one of main drawbacks of the original scheme and for reasons outlined in Section \ref{sec:intro}, deploying this solution in practice may be difficult.

The subpacketization issue was first discussed in the work of \cite{shanmugam_et_al14, shanmugam_et_al16} in the context of decentralized caching. Specifically, \cite{shanmugam_et_al16} showed that in the decentralized setting for any subpacketization level $F_s$ such that $F_s \leq \exp (KM/N)$ the rate would scale linearly in $K$, i.e., $R \geq cK$. Thus, much of the rate benefits of coded caching would be lost if $F_s$ did not scale exponentially in $K$. Following this work, the authors in \cite{yan_et_al17} introduced a technique for designing low subpacketization schemes in the centralized setting which they called placement delivery arrays. In \cite{yan_et_al17}, they considered the setting when $M/N = 1/q$ or $M/N = 1 - 1/q$ and demonstrated a scheme where the subpacketization level was exponentially smaller than the original scheme, while the rate was marginally higher. This scheme can be viewed as a special case of our work. We discuss these aspects in more detail in Section \ref{sec:compare}. In \cite{shangguan2016centralized}, the design of coded caching schemes was achieved through the design of hypergraphs with appropriate properties. In particular, for specific problem parameters, they were able to establish the existence of schemes where the subpacketization scaled as $\exp (c \sqrt{K})$. Reference\cite{yan2016placement} presented results in this setting by considering strong edge coloring of bipartite graphs.

Very recently, \cite{shanmugam2017coded} showed the existence of coded caching schemes where the subpacketization grows linearly with the number of users, but the coded caching rate grows as $O(K^\delta)$ where $0 < \delta < 1$. Thus, while the rate is not a constant, it does not grow linearly with $K$ either. Both \cite{shangguan2016centralized} and \cite{shanmugam2017coded} are interesting results that demonstrate the existence of regimes where the subpacketization scales in a manageable manner. Nevertheless, it is to be noted that these results come with several caveats. For example, the result of \cite{shanmugam2017coded} is only valid in the regime when $K$ is very large and is unlikely to be of use for practical values of $K$. The result of \cite{shangguan2016centralized} has significant restrictions on the number of users, e.g., in their paper, $K$ needs to be of the form $\binom{n}{a}$ and $q^t\binom{n}{a}$.

\subsection{Summary of Contributions}
In this work, the subpacketization levels we obtain are typically exponentially smaller than the original scheme. However, they still continue to scale exponentially in $K$, albeit with much smaller exponents. However, our construction has the advantage of being applicable for a large range of problem parameters.
Our specific contributions include the following.
\begin{itemize}
\item We uncover a simple and natural relationship between a $(n,k)$ linear block code and a coded caching scheme. We first show that any linear block code over $GF(q)$ and in some cases $\mathbb{Z} \mod q$ (where $q$ is not a prime or a prime power) generates a resolvable design. This design in turn specifies a coded caching scheme with $K=nq$ users where the cache fraction $M/N = 1/q$. A complementary cache fraction point where $M/N = 1 - \alpha/nq$ where $\alpha$ is some integer between $1$ and $k+1$ can also be obtained. Intermediate points can be obtained by memory sharing between these points.
\item We consider a class of $(n,k)$ linear block codes whose generator matrices satisfy a specific rank property. In particular, we require collections of consecutive columns to have certain rank properties. For such codes, we are able to identify an efficient delivery phase and determine the precise coded caching rate. We demonstrate that the subpacketization level is at most $q^k (k+1)$ whereas the coded caching gain scales as $k+1$ with respect to an uncoded caching scheme. Thus, different choices of $k$ allow the system designer significant flexibility to choose the appropriate operating point.
\item We discuss several constructions of generator matrices that satisfy the required rank property. We characterize the ranges of alphabet sizes $(q)$ over which these matrices can be constructed. If one has a given subpacketization budget in a specific setting, we are able to find a set of schemes that fit the budget while leveraging the rate gains of coded caching.
\end{itemize}
\section{Proposed low subpacketization level scheme}
 \label{sec:design}

All our constructions of low subpacketization schemes will stem from resolvable designs ({\it cf.} Definition \ref{def:resolv_design}). Our overall approach is to first show that any $(n,k)$ linear block code over $GF(q)$ can be used to obtain a resolvable block design. The placement scheme obtained from this resolvable design is such that $M/N = 1/q$. Under certain (mild) conditions on the generator matrix we show that a delivery phase scheme can be designed that allows for a significant rate gain over the uncoded scheme while having a subpacketization level that is significantly lower than \cite{maddahN14}. Furthermore, our scheme can be transformed into another scheme that operates at the point $M/N = 1 - \frac{k+1}{nq}$. Thus, intermediate values of $M/N$ can be obtained via memory sharing. We also discuss situations under which we can operate over modular arithmetic $\mathbb Z_q = \mathbb{Z} \mod q$ where $q$ is not necessarily a prime or a prime power; this allows us to obtain a larger range of parameters.

\subsection{Resolvable Design Construction}
\label{sec:construction}
 Consider a $(n,k)$ linear block code over $GF(q)$. To avoid trivialities we assume that its generator matrix does not have an all-zeros column. We collect its $q^k$ codewords 
 and construct a matrix $\bfT$ of size $n\times q^{k}$ as follows.
 \begin{equation}
 \label{eq:T}
 \bfT=[\bfc_0^T,\bfc_1^T,\cdots,\bfc_{q^{k}-1}^T],
 \end{equation}
where the $1 \times n$ vector $\bfc_\ell$ represents the $\ell$-th codeword of the code. Let $X=\{0,1,\cdots,q^{k}-1\}$ be the point set and $\calA$ be the collection of all subsets $B_{i,l}$ for $0\le i\le n-1$ and $0\le l\le q-1$, where
 \begin{align*}
 B_{i,l}&=\{j:\bfT_{i,j}=l\}.
 \end{align*}
Using this construction, we can obtain the following result.
\begin{lemma}
 \label{lemma:Resolve_MDS}
The construction procedure above results in a design $(X, \calA)$ where $X = \{0,1,\cdots, q^k-1\}$ and $|B_{i,l}| = q^{k-1}$ for all $0 \leq i \leq n-1$ and $0 \leq l \leq q-1$. Furthermore, the design is resolvable with parallel classes given by $\calP_i = \{B_{i,l}: 0 \leq l \leq q-1\}$, for $0\leq i \leq n-1$.
\end{lemma}
\begin{proof}
Let $\bfG=[g_{ab}]$, for $0\le a\le k-1$, $0\le b\le n-1$, $g_{ab}\in GF(q)$. Note that for $\Delta = [\Delta_0 ~\Delta_1~ \dots~\Delta_{n-1}]= \bfu \bfG$, we have
\begin{align*}
\Delta_b &= \sum_{a=0}^{k-1} \bfu_ag_{ab},
\end{align*}
where $\bfu=[\bfu_0,\cdots, \bfu_{k-1}]$. Let $a^*$ be such that $g_{a^*b} \neq 0$. Consider the equation
\begin{align*}
\sum_{a \neq a^*} \bfu_ag_{ab}&=\Delta_b-\bfu_{a^*}g_{a^*b},
\end{align*}
where $\Delta_b$ is fixed. For arbitrary values of $\bfu_a$, $a \neq a^*$, this equation has a unique solution for $\bfu_{a^*}$, which implies that for any $\Delta_b$, $|B_{b,\Delta_b   }| = q^{k-1}$ and that $\calP_{b}$ forms a parallel class.
\end{proof}
\begin{remark} A $k \times n$ generator matrix over $GF(q)$ where $q$ is a prime power can also be considered as a matrix over an extension field $GF(q^m)$ where $m$ is an integer. Thus, one can obtain a resolvable design in this case as well; the corresponding parameters can be calculated in an easy manner.
\end{remark}

\begin{remark} We can also consider linear block codes over $\mathbb{Z} \mod q$ where $q$ is not necessarily a prime or a prime power. In this case the conditions under which a resolvable design can be obtained by forming the matrix $\bfT$ are a little more involved. We discuss this in Lemma \ref{lemma:resol_ring} in the Appendix.
\end{remark}

\begin{example}
 	\label{eg:resolv_q_3_k_2}
 	Consider a $(4,2)$ linear block code over $GF(3)$ with generator matrix
 	$$
 	\bfG=\begin{bmatrix}
 	1&0&1&1\\
 	0&1&1&2
 	\end{bmatrix}.
 	$$
 	
 	Collecting the nine codewords, $\bfT$ is constructed as follows.
 	$$
 	\bfT=
 	\begin{bmatrix}
 	0&0&0&1&1&1&2&2&2\\
 	0&1&2&0&1&2&0&1&2\\
 	0&1&2&1&2&0&2&0&1\\
 	0&2&1&1&0&2&2&1&0
 	\end{bmatrix}.
 	$$
 	Using $\bfT$, we generate the resolvable block design $(X,\calA)$ where
 	the point set is $X=\{0,1,2,3,4,5,6,7,8\}$. For instance, block $B_{0,0}$ is obtained by identifying the column indexes of zeros in the first row of $\bfT$, i.e., $B_{0,0}=\{0,1,2\}$. Following this, we obtain
 	\begin{align*}
 	\calA=&\{\{0,1,2\},\{3,4,5\},\{6,7,8\},\{0,3,6\},\{1,4,7\},\{2,5,8\},\\
 	&\{0,5,7\}, \{1,3,8\},\{2,4,6\},\{0,4,8\},\{2,3,7\},\{1,5,6\}\}.
 	\end{align*}
 	It can be observed that $\calA$ has a resolution ({\it cf.} Definition \ref{def:resolv_design}) with the following parallel classes.
 	\begin{align*}
 	\calP_0 &=\{\{0,1,2\},\{3,4,5\},\{6,7,8\}\},\\
 	\calP_1 &=\{\{0,3,6\},\{1,4,7\},\{2,5,8\}\},\\
 	\calP_2 &=\{\{0,5,7\},\{1,3,8\},\{2,4,6\}\}, \text{~and}\\
 	\calP_3 &=\{\{0,4,8\},\{2,3,7\},\{1,5,6\}\}.
 	\end{align*}
 \end{example}
%
%

\subsection{A special class of linear block codes}

We now introduce a special class of linear block codes whose generator matrices satisfy specific rank properties. It turns out that resolvable designs obtained from these codes are especially suited for usage in coded caching.

Consider the generator matrix $\bfG$ of a $(n,k)$ linear block code over $GF(q)$. 
 The $i$-th column of $\bfG$ is denoted by $\bfg_i$.
Let $z$ be the least positive integer such that $k+1$ divides $nz$ (denoted by $k+1~|~nz$). We let $(t)_n$ denote $t~\text{mod}~n$.

In our construction we will need to consider various collections of $k+1$ consecutive columns of $\bfG$ (wraparounds over the boundaries are allowed). For this purpose, let $\calT_a=\{a(k+1),\cdots, a(k+1)+k\}$  ($a$ is a non-negative integer) and $\calS_a= \{(t)_n~|~t\in \calT_a \}$. Let $\bfG_{\calS_a}$ be the $k \times (k+1)$ submatrix of $\bfG$ specified by the columns in $\calS_a$, i.e.,  $\bfg_{\ell}$ is a column in  $\bfG_{\calS_a}$ if $\ell \in \calS_a$.
Next, we define the $(k,k+1)$-consecutive column property that is central to the rest of the discussion.

\begin{definition}{\it $(k,k+1)$-consecutive column property.}
\label{def:MDSproperty}
Consider the submatrices of $\bfG$ specified by  $\bfG_{\calS_a}$ for $0\le a\le \frac{zn}{k+1}-1$. We say that $\bfG$ satisfies the $(k,k+1)$-consecutive column property if all $k \times k$ submatrices of each $\bfG_{\calS_a}$ are full rank.
\end{definition}
Henceforth, we abbreviate the $(k,k+1)$-consecutive column property as $(k,k+1)$-CCP.
\begin{example}
\label{eg:code_k=2_n=4}
In Example \ref{eg:resolv_q_3_k_2} we have $k=2, n=4$ and hence $z=3$. Thus, $\calS_0 = \{0,1,2\}, \calS_1 = \{3,0,1\}, \calS_2 = \{2, 3, 0\}$ and $\calS_3 = \{1,2,3\}$.
The corresponding generator matrix $\bfG$ satisfies the $(k,k+1)$ CCP as any two columns of the each of submatrices $\bfG_{\calS_i}, i = 0, \dots, 3$ are linearly independent over $GF(3)$.
\end{example}

We note here that one can also define different levels of the consecutive column property. Let $\calT_a^\alpha=\{a\alpha,\cdots, a\alpha+\alpha-1\}$, $\calS_a^\alpha= \{(t)_n~|~t\in \calT_a^\alpha \}$ and $z$ is the least positive integer such that $\alpha| nz$.
\begin{definition}{\emph{$(k,\alpha)$-consecutive column property}}
\label{def:kalphacc}
Consider the submatrices of $\bfG$ specified by $\bfG_{\calS^\alpha_a}$ for $0\le a\le \frac{zn}{\alpha}-1$. We say that $\bfG$ satisfies the $(k,\alpha)$-consecutive column property, where $\alpha \leq k$ if each $\bfG_{\calS^\alpha_a}$ has full rank. In other words, the $\alpha$ columns in each $\bfG_{\calS^\alpha_a}$ are linearly independent.
\end{definition}
As pointed out in the sequel, codes that satisfy the $(k,\alpha)$-CCP, where $\alpha \leq k$ will result in caching systems that have a multiplicative rate gain of $\alpha$ over an uncoded system. 
Likewise, codes that satisfy the $(k,k+1)$-CCP will have a gain of $k+1$ over an uncoded system. In the remainder of the paper, we will use the term CCP to refer to the $(k,k+1)$-CCP if the value of $k$ is clear from the context.

\subsection{Usage in a coded caching scenario}
A resolvable design generated from a linear block code that satisfies the CCP can be used in a coded caching scheme as follows. We associate the users with the blocks. Each subfile is associated with a point and an additional index. The placement scheme follows the natural incidence between the blocks and the points; a formal description is given in Algorithm \ref{Alg:Placement} and illustrated further in Example \ref{eg:placement_q_3_k_2}. 

\begin{algorithm}[t]
	\caption{Placement Scheme}
	\label{Alg:Placement}
   \SetKwInOut{Input}{Input}
   \SetKwInOut{Output}{Output}
   \Input{Resolvable design $(X, \calA)$ constructed from a $(n,k)$ linear block code. Let $z$ be the least positive integer such that $k+1 ~|~ nz$.}
   Divide each file $W_n$, for $n \in [N]$ into $q^k z$ subfiles. Thus, $W_n = \{W_{n,t}^s : t \in \{0, \dots, q^k - 1\} \text{~and~} s \in \{0, \dots, z-1\}\}$ \;
   User $U_B$ for $B \in \calA$ caches $Z_B = \{W_{n,t}^s : n \in [N], t \in B \text{~and~} s \in \{0, \dots, z-1\} \}$ \;
   \Output{Cache content of user $U_B$ denoted $Z_B$ for $B \in \calA$.}
\end{algorithm}

\begin{example}
\label{eg:placement_q_3_k_2}
Consider the resolvable design from Example \ref{eg:resolv_q_3_k_2}, where we recall that $z=3$. The blocks in $\calA$ correspond to twelve users $U_{012}$, $U_{345}$, $U_{678}$, $U_{036}$, $U_{147}$, $U_{258}$, $U_{057}$, $U_{138}$, $U_{246}$, $U_{048}$, $U_{237}$, $U_{156}$. Each file is partitioned into $F_s=9\times z=27$ subfiles, each of which is denoted by $W_{n,t}^s$, $t=0,\cdots,8$, $s=0,1,2$. 
The cache in user $U_{abc}$, denoted $Z_{abc}$ is specified as $Z_{abc}=\{W_{n,t}^s~|~t \in \{a,b,c\}, s \in \{0,1,2\} \text{~and~} n \in [N]\}$. 
This corresponds to a coded caching system where each user caches $1/3$-rd of each file so that $M/N = 1/3$.
\end{example}

In general, (see Algorithm \ref{Alg:Placement}) we have $K = |\calA| = nq$ users. Each file $W_n$, $n\in [N]$ is divided into $q^kz$ subfiles $W_n=\{W_{n,t}^s~|~0\le t\le q^k-1,0\le s\le z-1\}$. A subfile $W_{n,t}^s$ is cached in user $U_B$ where $B \in \calA$ if $t\in B$. Therefore, each user caches a total of $Nq^{k-1}z$ subfiles. As each file consists of $q^k z$ subfiles, we have that $M/N = 1/q$.

It remains to show that we can design a delivery phase
scheme that satisfies any possible demand pattern.
Suppose that in the delivery phase user $U_B$ requests file $W_{d_{B}}$ where $d_B \in [N]$. The server responds by transmitting several equations that satisfy each user. Each equation allows $k+1$ users from {\it different parallel classes}  to simultaneously obtain a missing subfile. Our delivery scheme is such that the set of transmitted equations can be classified into various {\it recovery sets} that correspond to appropriate collections of parallel classes. For example, in Fig. \ref{Fig:Recovery}, $\calP_{\calS_0} = \{\calP_0, \calP_1, \calP_2\}, \calP_{\calS_1} = \{\calP_0, \calP_1, \calP_3\}$ and so on. It turns out that these recovery sets correspond precisely to the sets $\calS_a, 0 \leq a \leq \frac{zn}{k+1} -1$ defined earlier. We illustrate this by means of the example below.

\begin{figure}[t]
		\centering
		\includegraphics[scale=0.8]{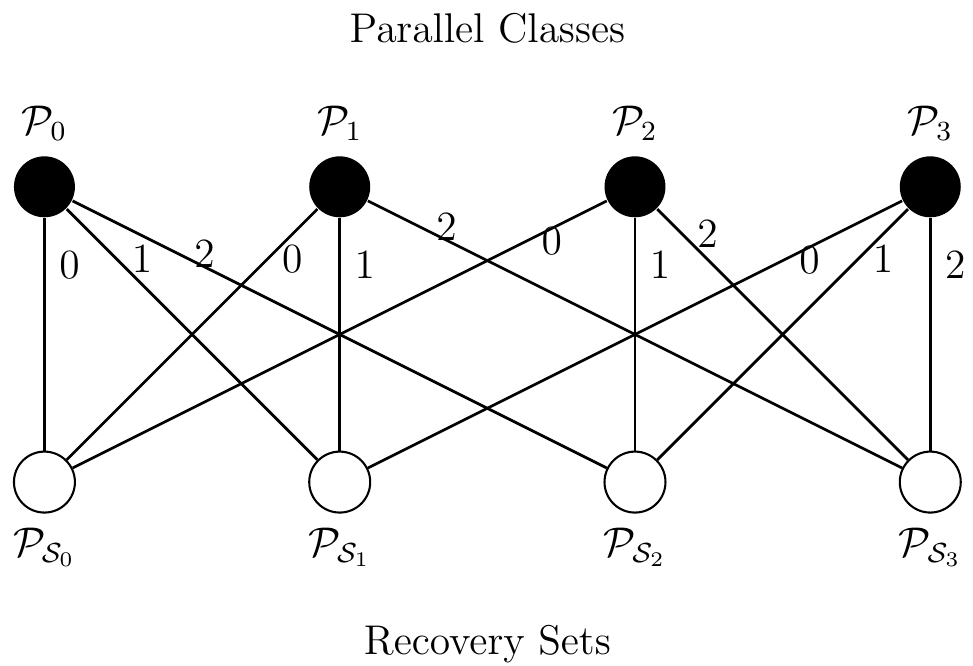}
		\caption{Recovery set bipartite graph}
		\label{Fig:Recovery}
\end{figure}

\begin{example}
	\label{ex:delivery}	
Consider the placement scheme specified in Example \ref{eg:placement_q_3_k_2}. Let each user $U_B$ request file $W_{d_B}$. The recovery sets are specified by means of the recovery set bipartite graph shown in Fig. \ref{Fig:Recovery}, e.g., $\calP_{\calS_1}$ corresponds to $\calS_1 = \{0,1,3\}$. The outgoing edges from each parallel class are labeled arbitrarily with numbers $0, 1$ and $2$. Our delivery scheme is such that each user recovers missing subfiles with a specific superscript from each recovery set that its corresponding parallel class participates in. For instance, a user in parallel class $\calP_1$ recovers missing subfiles with superscript $0$ from $\calP_{\calS_0}$, superscript 1 from $\calP_{\calS_1}$ and superscript 2 from $\calP_{\calS_3}$; these superscripts are the labels of outgoing edges from $\calP_1$ in the bipartite graph.


It can be verified, e.g., that user $U_{012}$ which lies in $\calP_0$ recovers all missing subfiles with superscript $1$ from the equations below.
	\begin{align*}
		&W^1_{d_{012},3}\oplus W^1_{d_{036},2}\oplus W^0_{d_{237},0},&
		W^1_{d_{012},6}\oplus W^1_{d_{036},1}\oplus W^0_{d_{156},0},\\	
		&W^1_{d_{012},4}\oplus W^1_{d_{147},0}\oplus W^0_{d_{048},1},&
		W^1_{d_{012},7}\oplus W^1_{d_{147},2}\oplus W^0_{d_{237},1},\\
		&W^1_{d_{012},8}\oplus W^1_{d_{258},0}\oplus W^0_{d_{048},2},&
		W^1_{d_{012},5}\oplus W^1_{d_{258},1}\oplus W^0_{d_{156},2}.
	\end{align*}
Each of the equations above benefits three users. They are generated simply by choosing $U_{012}$ from $\calP_0$, {\it any} block from $\calP_1$ and the last block from $\calP_3$ so that the {\it intersection of all these blocks is empty}. The fact that these equations are useful for the problem at hand is a consequence of the CCP. The process of generating these equations can be applied to all possible recovery sets. It can be shown that this allows all users to be satisfied at the end of the procedure.
	
\end{example}


In what follows, we first show that for the recovery set $\calP_{\calS_a}$ it is possible to generate equations that benefit $k+1$ users simultaneously.


\begin{claim}
\label{claim:MDSintersection}
	Consider the resolvable design $(X, \calA)$ constructed as described in Section III.A by a $(n,k)$ linear block code that satisfies the CCP. Let $\calP_{\calS_a}=\{\calP_i~|~i\in \calS_a\}$ for $0\le a\le \frac{zn}{k+1}-1$, i.e., it is the subset of parallel classes corresponding to $\calS_a$. We emphasize that $|\calP_{\calS_a}| = k+1$. Consider blocks $B_{i_1, l_{i_1}}, \dots, B_{i_{k}, l_{i_{k}}}$ (where $l_{i_j} \in \{0, \dots, q-1\}$) that are picked from any $k$ distinct parallel classes of $\calP_{\calS_a}$. Then, $|\cap_{j=1}^{k} B_{i_j, l_{i_j}}| = 1$.
\end{claim}
Before proving Claim \ref{claim:MDSintersection}, we discuss its application in the delivery phase. Note that the claim asserts that $k$ blocks chosen from $k$ distinct parallel classes intersect in precisely one point. Now, suppose that one picks $k+1$ users from $k+1$ distinct parallel classes, such that their intersection is empty. These blocks (equivalently, users) can participate in an equation that benefits $k+1$ users. In particular, each user will recover a missing subfile indexed by the intersection of the other $k$ blocks. We emphasize here that Claim \ref{claim:MDSintersection} is at the core of our delivery phase. Of course, we need to justify that enough equations can be found that allow all users to recover all their missing subfiles. This follows from a natural counting argument that is made more formally in the subsequent discussion. The superscripts $s \in \{0, \dots, z-1\}$ are needed for the counting argument to go through.
\begin{proof}
Following the construction in Section III.A, we note that a block $B_{i,l} \in \calP_i$ is specified by
$$
B_{i,l} = \{j : \bfT_{i,j} = l\}.
$$

Now consider $B_{i_1, l_{i_1}}, \dots, B_{i_{k}, l_{i_k}}$ (where $i_j \in \calS_a, l_{i_j} \in \{0, \dots, q-1\}$) that are picked from $k$ distinct parallel classes of  $\calP_{\calS_a}$. W.l.o.g. we assume that $i_1 < i_2 < \dots < i_{k}$. Let $\calI =  \{i_1, \dots, i_{k}\}$ and $\bfT_{\calI}$ denote the submatrix of $\bfT$ obtained by retaining the rows in $\calI$. We will show that the vector $[l_{i_1}~l_{i_2}~\dots~l_{i_k}]^T$ is a column in $\bfT_{\calI}$ and only appears once.

To see this  consider the system of equations in variables $\bfu_0, \dots, \bfu_{k-1}$.
\begin{align*}
\sum_{b=0}^{k-1}\bfu_{b}g_{bi_1} &= l_{i_1},\\
\mathrel{\makebox[\widthof{=}]{\vdots}}\\
\sum_{b=0}^{k-1}\bfu_{b}g_{bi_k} &= l_{i_k}.
\end{align*}
By the CCP, the vectors $\bfg_{i_1}, \bfg_{i_2}, \ldots, \bfg_{i_k}$ are linearly independent. Therefore this system of $k$ equations in $k$ variables has a unique solution over $GF(q)$. The result follows.
\end{proof}


\begin{algorithm}[t]
\SetNoFillComment
	\caption{Signal Generation Algorithm for $\calP_{\calS_a}$}
	\label{Alg:Signal}
   \SetKwInOut{Input}{Input}
   \SetKwInOut{Output}{Output}
   \Input{For $\calP \in \calP_{\calS_a}$, $E(\calP) = \text{label}(\calP - \calP_{\calS_a})$. Signal set $Sig=\emptyset$.}
   \While{any user $U_B\in \calP_j, j\in \calS_a$ does not recover all its missing subfiles with superscript $E(\calP_j)$}
   {Pick blocks $B_{j,l_j} \in \calP_j$ for all $j \in \calS_a$ and $l_j \in \{0, \dots, q-1\}$ such that $\cap_{j \in \calS_a} B_{j,l_j} = \emptyset$\; \tcc{Pick blocks from distinct parallel classes in $\calP_{\calS_a}$ such that their intersection is empty}
   	Let $\hat{l}_{s} = \cap_{j \in \calS_a \setminus \{s\}}B_{j,l_j}$ for $s\in \calS_a$\;
   \tcc{Determine the missing subfile index that the user from $\calP_s$ will recover}
   Add signal $\oplus_{s \in S_a} W^{E(\calP_s)}_{\kappa_{s,l_s},\hat{l}_s}$ to $Sig$
   \tcc{User $U_{B_{s,l_s}}$ demands file $W_{\kappa_{s,l_s}}$. This equation allows it to recover the corresponding missing subfile index $\hat{l}_s$. The superscript is determined by the recovery set bipartite graph}
   }

   \Output{Signal set $Sig$.}
\end{algorithm}
%
%
%
We now provide an intuitive argument for the delivery phase. Recall that we form a recovery set bipartite graph (see Fig. \ref{Fig:Recovery} for an example) with parallel classes and recovery sets as the disjoint vertex subsets. The edges incident on each parallel class are labeled arbitrarily from $0, \dots, z-1$. For a parallel class $\calP \in \calP_{\calS_a}$ we denote this label by $\text{label}(\calP - \calP_{\calS_a})$. For a given recovery set $\calP_{\calS_a}$, the delivery phase proceeds by choosing blocks from distinct parallel classes in $\calP_{\calS_a}$ such that their intersection is empty; this provides an equation that benefits $k+1$ users.
It turns out that the equation allows a user in parallel class $\calP \in \calP_{\calS_a}$ to recover a missing subfile with the superscript $\text{label}(\calP - \calP_{\calS_a})$.

The formal argument is made in Algorithm \ref{Alg:Signal}. For ease of notation in Algorithm \ref{Alg:Signal}, we denote the demand of user $U_{B_{i,j}}$ for $0 \leq i \leq n-1,0 \leq j \leq q-1$ by $W_{\kappa_{i,j}}$.

\begin{claim}
\label{claim:signalgeneration}
Consider a user $U_B$ belonging to parallel class $\calP \in \calP_{\calS_a}$. The signals generated in  Algorithm \ref{Alg:Signal} can recover all the missing subfiles needed by $U_B$
with superscript $E(\calP)$.
\end{claim}
\begin{proof}
    Let $\calP_\alpha \in \calP_{\calS_a}$. In the arguments below, we argue that user $U_{B_{\alpha,l_\alpha}}$ that demands file $W_{\kappa_{\alpha,l_\alpha}}$ can recover all its missing subfiles with superscript $E(\calP_\alpha)$.
    Note that $|B_{\alpha,l_\alpha}| = q^{k-1}$. Thus, user $U_{B_{\alpha,l_{\alpha}}}$ needs to obtain $q^k-q^{k-1}$ missing subfiles with superscript $E(\calP_\alpha)$.
    Consider an iteration of the while loop where block $B_{\alpha,l_\alpha}$ is picked in step 2.
    The equation in Algorithm \ref{Alg:Signal} allows it to recover $W^{E(\calP_\alpha)}_{\kappa_{\alpha,l_\alpha},\hat{l}_{\alpha}}$ where $\hat{l}_{\alpha} = \cap_{j \in \calS_a \setminus \{\alpha\}}B_{j,l_j}$.
    This is because $\cap_{j \in \calS_a} B_{j,l_j} = \emptyset$ and because of Claim \ref{claim:MDSintersection}.

    	
    Next we count the number of equations that $U_{B_{\alpha,l_{\alpha}}}$ participates in. We can pick $k-1$ users from some $k-1$ distinct parallel classes in $\calP_{\calS_a}$. This can be done in $q^{k-1}$ ways. Claim \ref{claim:MDSintersection} ensures that the blocks so chosen intersect in a single point. Next we pick a block from the only remaining parallel class in $\calP_{\calS_a}$ such that the intersection of all blocks is empty. This can be done in $q-1$ ways. Thus, there are a total of $q^{k-1}(q-1) = q^k - q^{k-1}$ equations in which user $U_{B_{\alpha,l_\alpha}}$ participates in.
    	
    It remains to argue that each equation provides a distinct subfile. Towards this end, let $\{i_1, \dots, i_k\} \subset \calS_a$ be an index set such that $\alpha \notin \{i_1, \dots, i_k\}$. Suppose that there exist sets of blocks $\{B_{i_1,l_{i_1}},\dots, B_{i_{k},l_{i_{k}}}\}$ and $\{B_{i_1,l'_{i_1}}, \dots, B_{i_{k},l'_{i_{k}}}\}$ such that $\{B_{i_1,l_{i_1}}, \dots, B_{i_{k},l_{i_{k}}}\}\neq \{B_{i_1,l'_{i_1}}, \dots, B_{i_{k},l'_{i_{k}}}\}$, but $\cap_{j=1}^{k} B_{i_j,l_{i_j}}=\cap_{j=1}^{k} B_{i_j,l'_{i_j}}=\beta$. This is a contradiction since this in turn implies that $\beta \in \cap_{j=2}^{k+1} B_{i_j,l_{i_j}}\bigcap \cap_{j=2}^{k+1} B_{i_j,l'_{i_j}}$, which is impossible since two blocks from the same parallel class have an empty intersection.

    As the algorithm is symmetric with respect to all blocks in parallel classes belonging to $\calP_{\calS_a}$, we have the required result.
\end{proof}

The overall delivery scheme  repeatedly applies Algorithm \ref{Alg:Signal} to each of the recovery sets.
\begin{lemma}
	\label{lemma:delivery}
	The proposed delivery scheme terminates and allows each user's demand to be satisfied. Furthermore the transmission rate of the server is $\frac{(q-1)n}{k+1}$ and the subpacketization level is $q^kz$.
\end{lemma}
\begin{proof}
See Appendix.
\end{proof}
The main requirement for Lemma \ref{lemma:delivery} to hold is that the recovery set bipartite graph be biregular, where multiple edges between the same pair of nodes is disallowed and the degree of each parallel class is $z$. It is not too hard to see that this follows from the definition of the recovery sets (see the proof in the Appendix for details).

In an analogous manner, if one starts with the generator matrix of a code that satisfies the $(k,\alpha)$-CCP for $\alpha \leq k$, then we can obtain the following result which is stated below. The details are quite similar to the discussion for the $(k,k+1)$-CCP and can be found in the Appendix (Section \ref{sec:kalpha_ccp_matrices}).
\begin{corollary}
Consider a coded caching scheme obtained by forming the resolvable design obtained from a $(n,k)$ code that satisfies the $(k,\alpha)$-CCP where $\alpha \leq k$. Let $z$ be the least positive integer such that $\alpha ~|~nz$. Then, a delivery scheme can be constructed such that the transmission rate is $\frac{(q-1)n}{\alpha}$ and the subpacketization level is $q^kz$. 
\end{corollary}

\subsection{Obtaining a scheme for $M/N = 1 - \frac{k+1}{nq}$.}
The construction above works for a system where $M/N=1/q$. It turns out that this can be converted into a scheme for $\frac{M}{N} = 1 - \frac{k+1}{nq}$. Thus, any convex combination of these two points can be obtained by memory-sharing.

Towards this end, we note that the class of coded caching schemes considered here can be specified by an \emph{equation-subfile} matrix. This is inspired by the hypergraph formulation and the placement delivery array (PDA) based schemes for coded caching in \cite{shangguan2016centralized} and \cite{yan_et_al17}. Each equation is assumed to be of the all-but-one type, i.e., it is of the form $W_{d_{t_1},\calA_{j_1}} \oplus W_{d_{t_2},\calA_{j_2}} \oplus \cdots \oplus W_{d_{t_{m}},\calA_{j_{m}}}$ where for each $\ell \in [m]$, we have the property that user $U_{t_\ell}$ does not cache subfile $W_{n,\calA_{j_\ell}}$ but caches all subfiles $W_{n,\calA_{j_s}}$ where $\{j_s : s \in [m], s \neq \ell\}$. 

The coded caching system corresponds to a $\Delta\times F_s$ equation-subfile matrix $\bfS$ as follows. We associate each row of $\bfS$ with an equation and each column with a subfile. We denote the $i$-th row of $\bfS$ by $Eq_i$ and $j$-th column of $\bfS$ by $\calA_j$. The value $\bfS(i,j)=t$ if in the $i$-th equation, user $U_t$ recovers subfile $W_{d_t,\calA_j}$, otherwise, $\bfS(i,j)=0$. Suppose that these $\Delta$ equations allow each user to satisfy their demands, i.e., $\bfS$ corresponds to a valid coded caching scheme. It is not too hard to see that the placement scheme can be obtained by examining $\bfS$. Namely, user $U_t$ caches the subfile corresponding to the $j$-th column if integer $t$ does not appear in the $j$-th column.



\begin{example}
	\label{ex:lemmablock}
	Consider a coded caching system in \cite{maddahN14} with $K=4$, $\Delta=4$ and $F_s=6$. We denote the four users as $U_1, U_2, U_3, U_4$.
 Suppose that the equation-subfile matrix $\bfS$ for this scheme is as specified below.
	\begin{align*}
	\bordermatrix{~ & \calA_1 & \calA_2 & \calA_3 & \calA_4 & \calA_5 & \calA_6 \cr
		Eq_1 & 3 & 2 & 0 & 1 & 0 & 0\cr
		Eq_2 & 4 & 0 & 2 & 0 & 1 & 0\cr
		Eq_3 & 0 & 4 & 3 & 0 & 0 & 1\cr
		Eq_4 & 0 & 0 & 0 & 4 & 3 & 2\cr}.
	\end{align*}
Upon examining $\bfS$ it is evident for instance that user $U_1$ caches subfiles $\calA_1, \dots, \calA_3$ as the number $1$ does not appear in the corresponding columns. Similarly, the cache placement of the other users can be obtained. Interpreting this placement scheme in terms of the user-subfile assignment, it can be verified that the design so obtained corresponds to the transpose of the scheme considered in Example \ref{eq:eg_2_intro} (and also to the scheme of \cite{maddahN14} for $K=4$, $M/N = 1/2$).
 %
\end{example}

\begin{lemma}
	\label{lemma:symmcache}
	Consider a $\Delta\times F_s$ equation-subfile matrix $\bfS$ whose entries belong to the set $\{0, 1, \dots, K\}$. It corresponds to a valid coded caching system if the following three conditions are satisfied.
	\begin{itemize}
		\item There is no non-zero integer appearing more than once in each column.
		\item There is no non-zero integer appearing more than once in each row.
		\item If $\bfS(i_1,j_1)=\bfS(i_2,j_2)\neq 0$, then $\bfS(i_1,j_2)=\bfS(i_2,j_1)=0$.
	\end{itemize}

\end{lemma}
\begin{proof}
The placement scheme is obtained as discussed earlier, i.e., user $U_t$ caches subfiles $W_{n,\calA_j}$ if integer $t$ does not appear in column $\calA_j$. Therefore, matrix $\bfS$ corresponds to a placement scheme.

Next we discuss the delivery scheme. Note that $Eq_i$ corresponds to an equation as follows.
$$W_{d_{t_1},\calA_{j_1}} \oplus W_{d_{t_2},\calA_{j_2}} \oplus \cdots \oplus W_{d_{t_{m}},\calA_{j_{m}}},$$
where $\bfS(i,j_1)=t_1,\cdots, \bfS(i,j_{m})=t_{m}$. The above equation can allow $m$ users to recover subfiles simultaneously if (a) $U_{t_\ell}$ does not cache $W_{n,\calA_{j_\ell}}$ and (b) $U_{t_\ell}$ caches all $W_{n,\calA_{j_s}}$ where $\{j_s : s \in [m], s \neq \ell \}$. It is evident that $U_{t_\ell}$ does not cache $W_{n,\calA_{j_\ell}}$ owing to the placement scheme. Next, to guarantee the condition (b), we need to show that integer $t_\ell=\bfS(i,j_\ell)$ will not appear in column $\calA_{j_s}$ in $\bfS$ where $\{j_s : s \in [m], s \neq \ell\}$. Towards this end, $t_\ell \neq \bfS(i,j_s)$ because of Condition 2. Next, consider the non-zero entries that lie in the column $\calA_{j_s}$ but not in the row $Eq_i$. Assume there exists an entry $\bfS(i',j_s)$ such that $\bfS(i',j_s)=\bfS(i,j_\ell)=t_\ell$ and $i'\neq i$, then $\bfS(i,j_s)=t_s\neq 0$, which is a contradiction to Condition 3. Finally, Condition 1 guarantees that each missing subfile is recovered only once.
\end{proof}
User $U_t$ caches a fraction $\frac{M_t}{N}=\frac{L_t}{F_s}$ where $L_t$ is the number of columns of $\bfS$ that do not have the entry $t$. Similarly, the transmission rate is given by $R=\frac{\Delta}{F_s}$.



The crucial point is that the transpose of $\bfS$, i.e., $\bfS^T$ also corresponds to a coded caching scheme. This follows directly from the fact that $\bfS^T$ also satisfies the conditions in Lemma \ref{lemma:symmcache}. In particular, $\bfS^T$ corresponds to a coded caching system with $K$ users and $\Delta$ subfiles. In the placement phase, the cache size of $U_t$ is $\frac{M'_t}{N}=\frac{\Delta-F_s+L_t}{\Delta}$. In the delivery phase, by transmitting $F_s$ equations corresponding to the rows of $\bfS^T$, all missing subfiles can be recovered. Then, the transmission rate is $R'=\frac{F_s}{\Delta}$.

Applying the above discussion in our context, consider the equation-subfile matrix $\bfS$ corresponding to the coded caching system with $K=nq$, $\frac{M_t}{N}=\frac{1}{q}$ for $1\le t\le nq$, $F_s=q^kz$ and $\Delta= q^k(q-1)\frac{nz}{k+1}$. Then $\bfS^T$ corresponds to a system with $K'=nq$, $\frac{M'}{N}=1-\frac{k+1}{nq}$, $F'_s=(q-1)q^k\frac{zn}{k+1}$, and transmission rate $R'=\frac{F_s}{\Delta}=\frac{k+1}{(q-1)n}$. The following theorem is the main result of this paper.
\begin{theorem}
	\label{them:main}
	Consider a $(n,k)$ linear block code over $GF(q)$ that satisfies the $(k,k+1)$ CCP. This corresponds to a coded caching scheme with $K=nq$ users, $N$ files in the server where each user has a cache of size $M \in \bigg{\{} \frac{1}{q}N,\bigg{(}1-\frac{k+1}{nq}\bigg{)}N \bigg{\}}$. Let $z$ be the least positive integer such that $k+1~|~nz$. When $\frac{M}{N}=\frac{1}{q}$, we have
	\begin{align*}
		R&=\frac{(q-1)n}{k+1}, \text{~and}\\
		F_s&=q^kz.
	\end{align*}
	When $\frac{M}{N}=(1-\frac{k+1}{nq})$, we have
	\begin{align*}
		R&=\frac{k+1}{(q-1)n}, \text{~and}\\
		F_s&=(q-1)q^k\frac{zn}{k+1}.
	\end{align*}
By memory sharing any convex combination of these points is achievable.
\end{theorem}

In a similar manner for the $(n,k)$ linear block code that satisfies the $(k,\alpha)$-CCP over $GF(q)$, the caching system where $M/N=1/q$ can be converted into a system where $K'=nq$, $\frac{M'}{N'}=1-\frac{\alpha}{nq}$, $F_s'=(q-1)q^k\frac{zn}{\alpha}$ and $R'=\frac{\alpha}{(q-1)n}$ using the equation-subfile technique. The arguments presented above apply with essentially no change.



%
%

\section{Some classes of linear codes that satisfy the CCP}
\label{sec:LinearCode}
At this point we have established that linear block codes that satisfy the CCP are attractive candidates for usage in coded caching. In this section, we demonstrate that there are a large class of generator matrices that satisfy the CCP. For most of the section we work with matrices over a finite field of order $q$. In the last subsection, we discuss some constructions for matrices over $\mathbb{Z} \mod q$ when $q$ is not a prime or prime power. We summarize the constructions presented in this section in Table \ref{Table:construction}.

{
	\renewcommand{\arraystretch}{1.5}
	\begin{table*}[t]
	\begin{center}
		\begin{tabular}{|p{3cm}|p{5cm}|p{7cm}|}
			\hline
			\hline
			 {\bf Code type} & {\bf Code construction} & {\bf Notes}  \\
			\hline
			\hline
			\multirow{5}{*}{Codes over field $GF(q)$} & $(n,k)$ MDS codes  & Satisfy $(k,k+1)$-CCP. Need $q+1 \geq n$. \\
			\cline{2-3}
			 & $(n,k)$ Cyclic codes  & Existence depends on certain properties of the generator polynomials. All cyclic codes satisfy the $(k,k)$-CCP. Need additional conditions for the $(k,k+1)$-CCP.  \\
			\cline{2-3}
			  & Kronecker product of $z \times \alpha$ matrix satisfying the $(z,z)$-CCP with the identity matrix $\bfI_{t \times t}$ & Satisfy the $(k,k)$-CCP where $k=tz$.\\
			\cline{2-3}
			 & Kronecker product of Vandermonde and Vandermonde-like matrices with structured base matrices & Satisfy the $(k,k+1)$-CCP for certain parameters.\\
            \cline{2-3}
            & CCP matrix extension & Extends a $k \times n$ CCP matrix to a $k \times (n + s(k+1))$ CCP matrix for integer s.\\
			\hline
			\hline
			\multirow{3}{*}{Codes over ring $\mathbb{Z} \mod q$}  & Single parity-check (SPC) code  & Satisfy the $(k,k+1)$-CCP with $n=k+1$. \\
			\cline{2-3}
			   & Cyclic codes over the ring  & Require that $q = q_1 \times q_2 \times \dots \times q_d$ where $q_i$'s are prime. Satisfy the $(k,k)$-CCP.  \\
			\cline{2-3}
			  & Kronecker product of $z \times \alpha$ matrix satisfying the $(z,z)$-CCP with the identity matrix $\bfI_{t \times t}$ & Satisfy the $(k,k)$-CCP property where $k=tz$.\\
			\cline{2-3}
            & CCP matrix extension & Extends a $k \times n$ CCP matrix to a $k \times (n + s(k+1))$ CCP matrix for integer s.\\
			\hline
			\hline
		\end{tabular}
	\end{center}
	\caption{\label{Table:construction} A summary of the different constructions of CCP matrices in Section \ref{sec:LinearCode}}
\end{table*}
}

\subsection{Maximum-distance-separable (MDS) codes}
\label{sec:MDS_cons}
$(n,k)$-MDS codes with minimum distance $n-k+1$ are clearly a class of codes that satisfy the CCP. In fact, for these codes any $k$ columns of the generator matrix can be shown to be full rank. Note however, that MDS codes typically need large field size, e.g., $q +1 \geq n$  (assuming that the MDS conjecture is true)\cite{RonRoth}. In our construction, the value of $M/N = 1/q$ and the number of users is $K = nq$. Thus, for large $n$, we will only obtain systems with small values of $M/N$, or equivalently large values of $M/N$ (by Theorem \ref{them:main} above). This may be restrictive in practice.

\subsection{Cyclic Codes}
     A cyclic code is a linear block code, where the circular shift of each codeword is also a codeword \cite{lincostello}. A $(n,k)$ cyclic code over $GF(q)$ is specified by a monic polynomial $g(X) = \sum_{i=0}^{n-k} g_i X^i$ with coefficients from $GF(q)$ where $g_{n-k} = 1$ and $g_0 \neq 0$; $g(X)$ needs to divide the polynomial $X^n-1$. The generator matrix of the cyclic code is obtained as below.
     $$
     \bfG=\begin{bmatrix}
     g_0&g_1&\cdot&\cdot&\cdot&g_{n-k}&0&\cdot&\cdot&0\\
     0&g_0&g_1&\cdot&\cdot&\cdot&g_{n-k}&0&\cdot&0\\
     \vdots&&&&&&&&&\vdots\\
     0&0&\cdot&0&g_0&g_1&\cdot&\cdot&\cdot&g_{n-k}
     \end{bmatrix}.
     $$


The following claim shows that for verifying the CCP for a cyclic code it suffices to pick {\it any} set of $k+1$ consecutive columns.
\begin{claim}
	\label{claim:cyclic_k1}
	Consider a $(n,k)$ cyclic code with generator matrix $\bfG$. Let $\bfG_\calS$ denote a set of $k+1$ consecutive columns of $\bfG$. If each $k \times k$ submatrix of $\bfG_S$ is full rank, then $\bfG$ satisfies the $(k,k+1)$-CCP.
\end{claim}
\begin{proof}

Let the generator polynomial of the cyclic code be $\bfg(X)$, where we note that $\bfg(X)$ has degree $n-k$. Let $\bfG_{\calS}=[\bfg_{(a)_n},\bfg_{(a+1)_n},\cdots,\bfg_{(a+k)_n}]$ where we assume that $\bfG_{\calS}$ satisfies the $(k,k+1)$-CCP. 
Let
\begin{align*}
\bfG_{\calS\setminus j} =& [\bfg_{(a)_n},\dots,\bfg_{(a+j-1)_n},\bfg_{(a+j+1)_n},\dots,\bfg_{(a+k)_n}], \text{and}\\
\bfG_{\calS'\setminus j} =&[\bfg_{(a+i)_n},\dots,\bfg_{(a+j-1+i)_n},\\&~\bfg_{(a+j+1+i)_n},
\dots,\bfg_{(a+k+i)_n}].
\end{align*}
We need to show that if $\bfG_{\calS\setminus j}$ has full rank, then $\bfG_{\calS'\setminus j}$ has full rank, for any $0\le j\le k$. 

As $\bfG_{\calS\setminus j}$ has full rank, there is no codeword $\bfc \neq \mathbf{0}$ such that $\bfc((a)_n)=\cdots=\bfc((a+j-1)_n)=\bfc((a+j+1)_n)=\cdots=\bfc((a+k)_n)=0.$
By the definition of a cyclic code, any circular shift of a codeword results in another codeword that belongs to the code. Therefore, there is no codeword $\bfc'$ such that $\bfc'((a+i)_n)=\cdots=\bfc'(a+j-1+i)_n)=\bfc'((a+j+1+i)_n)=\cdots=\bfc'((a+k+i)_n)=0.$ Thus, $\bfG_{\calS'\setminus j}$ has full rank.
\end{proof}

Claim \ref{claim:cyclic_k1} implies a low complexity search algorithm to determine if a cyclic code satisfies the CCP. Instead of checking all $\bfG_{\calS_a}$, $0\le a\le \frac{zn}{k+1}-1$, in Definition \ref{def:MDSproperty}, we only need to check an arbitrary $\bfG_{\calS}=[\bfg_{(i)_n},\bfg_{(i+1)_n},\cdots,\bfg_{(i+k)_n}]$, for $0\le i< n$. 
To further simplify the search, we choose $i=n-\lfloor \frac{k}{2}\rfloor-1$.


For this choice of $i$, Claim \ref{claim:cyclic_MDS} shows that $\bfG_\calS$ is such that we only need to check the rank of a list of small-dimension matrices to determine if each $k\times k$ submatrix of $\bfG_\calS$ is full rank (the proof appears in the Appendix).

\begin{claim}
	\label{claim:cyclic_MDS}
	A cyclic code with generator matrix $\bfG$ satisfies the CCP if the following conditions hold.
\begin{itemize}
	\item
	For $0<j\le \lfloor \frac{k}{2}\rfloor$, the submatrices
	\begin{align*}
    \bfC_j =\begin{bmatrix}
	g_{n-k-1}&g_{n-k}&0&\cdot&\cdot&0\\
	g_{n-k-2}&g_{n-k-1}&g_{n-k}&0&\cdot&0\\
	\vdots&&&&&\vdots\\
	g_{n-k-j+1}&\cdot&\cdot&\cdot&\cdot&g_{n-k}\\
	g_{n-k-j}&\cdot&\cdot&\cdot&\cdot&g_{n-k-1}
	\end{bmatrix}
	\end{align*}
	have full rank. In the above expression, $g_i=0$ if $i<0$.
	\item
	For $\lfloor\frac{k}{2}\rfloor<j<k$, the submatrices
		\begin{align*}
		\bfC_j =\begin{bmatrix}
		g_{1}&g_{2}&\cdot&\cdot&\cdot&\cdot&g_{k-j}\\
		g_{0}&g_{1}&\cdot&\cdot&\cdot&\cdot&g_{k-j-1}\\
		\vdots&&&&&&\vdots\\
		0&\cdot&\cdot&0&g_0&g_1&g_2\\
		0&\cdot&&\cdot\cdot&0&g_0&g_1
		\end{bmatrix}
		\end{align*}
	have full rank.
\end{itemize}
\end{claim}

\begin{example}
	Consider the polynomial $\bfg(X)=X^4+X^3+X+2$ over $GF(3)$. Since it divides $X^8-1$, it is the generator polynomial of a $(8,4)$ cyclic code over $GF(3)$. The generator matrix of this code is given below.
	$$\bfG=
	\begin{bmatrix}
	2&1&0&1&1&0&0&0\\
	0&2&1&0&1&1&0&0\\
	0&0&2&1&0&1&1&0\\
	0&0&0&2&1&0&1&1
	\end{bmatrix}.
	$$
	It can be verified that the $4\times 5$ submatrix which consists of the two leftmost columns and three rightmost columns of $\bfG$ is such that all $4 \times 4$ submatrices of it are full rank. Thus, by Claim \ref{claim:cyclic_k1} the (4,5)-CCP is satisfied for $\bfG$.
\end{example}
\begin{remark}
Cyclic codes form an important class of codes that satisfy the $(k,k)$-CCP ({\it cf.} Definition \ref{def:kalphacc}). This is because, it is well-known \cite{lincostello} that any $k$ consecutive columns of the generator matrix of a cyclic code are linearly independent. 
\end{remark}
\subsection{Constructions leveraging properties of smaller base matrices}
\label{sec:more_cons_block}
It is well recognized that cyclic codes do not necessarily exist for any choice of parameters. This is because of the divisibility requirement on the generator polynomial. We now discuss a more general construction of generator matrices that satisfy the CCP. As we shall see, this construction provides a more or less satisfactory solution for a large range of system parameters.


Our first simple observation is that the Kronecker product (denoted by $\otimes$ below) of a $z \times \alpha$ generator matrix that satisfies the $(z,z)$-CCP with a $t \times t$ identity matrix, $\bfI_{t\times t}$ immediately yields a generator matrix that satisfies the $(tz,tz)$-CCP.
%
\begin{claim}
	\label{claim:GconstructKKCCP}

	Consider a $(n,k)$ linear block code over $GF(q)$ whose generator matrix is specified as $\bfG=\bfA\otimes \bfI_{t\times t}$ where $\bfA$ is a $z\times \alpha$ matrix that satisfies the $(z,z)$-CCP. Then, $\bfG$ satisfies the $(k,k)$-CCP  where $k=tz$ and $n=t\alpha$.
\end{claim}
\begin{proof}
	The recovery set for $\bfA$ is specified as $\calS_a^z=\{(az)_\alpha,\cdots, (az+z-1)_\alpha\}$ and the recovery set for $\bfG$ is specified as $\calS_a^k=\{(ak)_n, \cdots, (ak+k-1)_n\}$. Since $\bfA$ satisfies the $(z,z)$-CCP, $\bfA_{\calS_a^z}$ has full rank.  Note that $\bfG_{\calS_a^k}=\bfA_{\calS_a^z}\otimes \bfI_{t\times t}$. Then $\det(\bfG_{\calS_a^k})=\det(\bfA_{\calS_a^z}\otimes \bfI_{t\times t})=\det(\bfA_{\calS_a^z})^t\neq 0$. Therefore, $\bfG$ satisfies the $(k,k)$-CCP. 
\end{proof}
\begin{remark}
Let $\bfA$ be the generator matrix of a cyclic code over $GF(q)$, then $\bfG=\bfA \otimes \bfI_{t\times t}$ satisfies the $(k,k)$-CCP by Claim \ref{claim:GconstructKKCCP}.
\end{remark}
Our next construction addresses the $(k,k+1)$-CCP. In what follows, we use the following notation.
\begin{itemize}
	\item $\bfoe_a$: $[\underbrace{1,\cdots,1}_a]^T$;
	\item $\bfC(c_1, c_2)_{a\times b}$: $a\times b$ matrix where each row is the cyclic
	shift (one place to the right) of the row above it and the first row is $[c_1~ c_2~ 0~ \cdots~ 0]$; and
	\item $\bfzr_{a\times b}$: $a\times b$ matrix with zero entries.
\end{itemize}

Consider parameters $n,k$. Let the greatest common divisor of $n$ and $k+1$, $\gcd(n,k+1)=t$. It is easy to verify that $z=\frac{k+1}{t}$ is the smallest integer such that $k+1~|~nz$. Let $n=t\alpha$ and $k+1=tz$. Claim \ref{claim:Gconstruct1} below constructs a $(n,k)$ linear code that satisfies the CCP over $GF(q)$ where $q > \alpha$. Since $\alpha=\frac{n}{t}$, the required field size in Claim \ref{claim:Gconstruct1} is lower than the MDS code considered in Section \ref{sec:MDS_cons}. 
\begin{claim}
	\label{claim:Gconstruct1}
	Consider a $(n,k)$ linear block code over $GF(q)$ whose generator matrix is specified as eq. (\ref{eq:Gconstruct1}),
	\begin{table*}[t]
			\begin{equation}
		\label{eq:Gconstruct1}
		\bfG=
		\begin{bmatrix}
		b_{00}\bfI_{t\times t}&b_{01}\bfI_{t\times t}&\cdots&b_{0(\alpha-1)}\bfI_{t\times t}\\
		b_{10}\bfI_{t\times t}&b_{11}\bfI_{t\times t}&\cdots&b_{1(\alpha-1)}\bfI_{t\times t}\\
		\vdots&&&\vdots\\
		b_{(z-2)0}\bfI_{t\times t}&b_{(z-2)1}\bfI_{t\times t}&\cdots&b_{(z-2)(\alpha-1)}\bfI_{t\times t}\\
		\bfC(b_{(z-1)0}, b_{(z-1)0})_{(t-1)\times t}&\bfC(b_{(z-1)1}, b_{(z-1)1})_{(t-1)\times t}&\cdots&\bfC(b_{(z-1)(\alpha-1)}, b_{(z-1)(\alpha-1)})_{(t-1)\times t}
		\end{bmatrix}
		\end{equation}
	\end{table*}
	where
	\begin{align*}
    \begin{bmatrix}
	    b_{00}&b_{01}&\cdots&b_{0(\alpha-1)}\\
    	b_{10}&b_{11}&\cdots&b_{1(\alpha-1)}\\
	    \vdots&&&\vdots\\
	    b_{(z-1)0}&b_{(z-1)1}&\cdots&b_{(z-1)(\alpha-1)}
	\end{bmatrix}
	\end{align*}
	is a Vandermonde matrix and $q> \alpha$. Then, $\bfG$ satisfies the $(k,k+1)$-CCP.
\end{claim}
\begin{proof}
The proof again leverages the idea that $\bfG$ can be expressed succinctly by using Kronecker products. The arguments can be found in the Appendix.
\end{proof}

Consider the case when $\alpha=z+1$. We construct a $(n,k)$ linear code satisfy the CCP over $GF(q)$ where $q\ge z$. It can be noted that the constraint of field size is looser than the corresponding constraint in Claim \ref{claim:Gconstruct1}. 
\begin{claim}
\label{claim:Gconstruct2}

Consider $nz=(z+1)\cdot (k+1)$ . Consider a $(n,k)$ linear block code whose generator matrix (over $GF(q)$) is specified as follows. 

\begin{equation}
\label{eq:Gconstruct2}
\bfG=
\begin{bmatrix}
\begin{smallmatrix}
\bfI_{t\times t}&\bfzr_{t\times t}&\cdots&\bfzr_{t\times t}&\bfzr_{t\times (t-1)}&\bfoe_t&{b_1}\bfI_{t\times t}\\
\bfzr_{t\times t}&\bfI_{t\times t}&\cdots&\bfzr_{t\times t}&\bfzr_{t\times (t-1)}&\bfoe_t&{b_2}\bfI_{t\times t}\\
\vdots&&&&&&\vdots\\
\bfzr_{t\times t}&\bfzr_{t\times t}&\cdots&\bfI_{t\times t}&\bfzr_{t\times (t-1)}&\bfoe_t&{b_{z-1}}\bfI_{t\times t}\\
\bfzr_{(t-1)\times t}&\bfzr_{(t-1)\times t}&\cdots&\bfzr_{(t-1)\times t}&\bfI_{(t-1)\times (t-1)}&\bfoe_{t-1}&\bfC(c_1, c_2)_{(t-1)\times t},
\end{smallmatrix}
\end{bmatrix}
\end{equation}
where $t=n-k-1$. If $q\ge z$, $b_1,b_2,\cdots,b_{z-1}$ are non-zero and distinct, and $c_1+c_2=0$, then $\bfG$ satisfies the CCP.

\end{claim}
\begin{proof}
	See Appendix.
\end{proof}

Given a $(n,k)$ code that satisfies the CCP, we can use it obtain higher values of $n$ in a simple manner as discussed in the claim below.
\begin{claim}
	\label{claim:matrixextend}
	Consider a $(n,k)$ linear block code over $GF(q)$ with generator matrix $\bfG$ that satisfies the CCP. Let the first $k+1$ columns of $\bfG$ be denoted by the submatrix $\bfD$. Then the matrix $\bfG'$ of dimension $k \times (n + s(k+1))$ where $s\ge 0$
	\begin{align*}
	\bfG'=[\underbrace{\bfD| \cdots |\bfD}_s| \bfG]
	\end{align*}
also satisfies the CCP.
\end{claim}
\begin{proof}
	See Appendix.
\end{proof}

Claim \ref{claim:matrixextend} can provide more parameter choices and more possible code constructions. For example, given $n,k,q$, where $k+1+(n)_{k+1}\le q+1<n$, there may not exist a $(n, k)$-MDS code over $GF(q)$. However, there exists a $(k+1+(n)_{k+1}, k)$-MDS code over $GF(q)$. By Claim \ref{claim:matrixextend}, we can obtain a $(n,k)$ linear block code over $GF(q)$ that satisfies the CCP. Similarly, combining Claim \ref{claim:cyclic_MDS}, Claim \ref{claim:Gconstruct1}, Claim \ref{claim:Gconstruct2} with Claim \ref{claim:matrixextend}, we can obtain more linear block codes that satisfy the CCP. 

A result very similar to Claim \ref{claim:matrixextend} can be obtained for the $(k,\alpha)$-CCP. Specifically, consider a $(n,k)$ linear block code with generator matrix $\bfG$ that satisfies the $(k,\alpha)$-CCP and let $\bfD$ be the first $\alpha$ columns of $\bfG$. Then, $\bfG'=[\underbrace{\bfD| \cdots |\bfD}_s| \bfG]$ of dimension $k\times (n+s\alpha)$ also satisfies the $(k,\alpha)$-CCP.



\subsection{Constructions where $q$ is not a prime or a prime power}
\label{sec:matrices_z_mod_q}
We now discuss constructions where $q$ is not a prime or a prime power. We attempt to construct matrices over the ring $\mathbb{Z} \mod q$ in this case. The issue is somewhat complicated by the fact that a square matrix over $\mathbb{Z} \mod q$ has linear independent rows if and only if its determinant is a unit in the ring \cite{dummit2003abstract}. In general, this fact makes it harder to obtain constructions such as those in Claim \ref{claim:Gconstruct1} that exploit the Vandermonde structure of the matrices. Specifically, the difference of units in a ring is not guaranteed to be a unit. However, we can still provide some constructions. It can be observed that Claim \ref{claim:GconstructKKCCP} and Claim \ref{claim:matrixextend} hold for linear block codes over $\mathbb{Z} \mod q$. We will use them without proof in this subsection. 


\begin{claim}
	\label{claim:SPC_CCPZmodq}
Let $\bfG = [\bfI_{k \times k} | \bfoe_k]$, i.e., it is the generator matrix of a $(k+1, k)$ single parity check (SPC) code, where the entries are from $\mathbb{Z} \mod q$.
The $\bfG$ satisfies the $(k,k+1)$-CCP and the $(k,k)$-CCP. It can be used as base matrix for  Claim \ref{claim:GconstructKKCCP}.
\end{claim}
\begin{proof}
It is not too hard to see that when $\bfG = [\bfI_{k \times k} | \bfoe_k]$, any $k \times k$ submatrix of $\bfG$ has a determinant which is $\pm 1$, i.e., it is a unit over $\mathbb{Z} \mod q$. Thus, the result holds in this case.
\end{proof}
\begin{claim}
\label{claim:z_2_Zmodq}
The following matrix with entries from $\mathbb{Z}\mod q$ satisfies the $(k,k+1)$-CCP. Here $k =2t-1$ and $n=3t$.
\begin{align*}
    \bfG = \begin{bmatrix}
	   \bfI_{t \times t}& \bfzr & \bfoe_t & \bfI_{t \times t}\\
    	 \bfzr &\bfI_{(t-1) \times (t-1)}& \bfoe_t & \bfC(1,-1)_{(t-1) \times t}
	\end{bmatrix}.
\end{align*}
\end{claim}
\begin{proof}
This can be proved by following the arguments in the proof of Claim \ref{claim:Gconstruct2} while treating elements to be from $\mathbb{Z} \mod q$ and setting $z=2$. We need to consider three different $k \times (k+1)$ submatrices for which we need to check the property. These correspond to simpler instances of the submatrices considered in Types I - III in the proof of Claim \ref{claim:Gconstruct2}. In particular, the corresponding determinants will always be $\pm 1$ which are units over $\mathbb{Z} \mod q$.
\end{proof}
\begin{remark}
We note that the general construction in Claim \ref{claim:Gconstruct2} can potentially fail in the case when the matrices are over $\mathbb{Z} \mod q$. This is because in one of the cases under consideration (specifically, Type III, Case 1), the determinant depends on the difference of the $b_i$ values. The difference of units in $\mathbb{Z} \mod q$ is not guaranteed to be a unit, thus there is no guarantee that the determinant is a unit.
\end{remark}
\begin{remark}
We can use Claim \ref{claim:matrixextend} to obtain higher values of $n$ based on the above two classes of linear block codes over $\mathbb{Z} \mod q$.
\end{remark}

%
%

While most constructions of cyclic codes are over $GF(q)$, there has been some work on constructing cyclic codes over $\mathbb{Z}\mod q$. Specifically, \cite{blake1972codes} provides a construction where $q=q_1\times q_2 \cdots \times q_d$ and $q_i, i = 1, \dots, d$ are prime. We begin by outlining this construction. By the Chinese remainder theorem any element $ \gamma \in \mathbb{Z}\mod q$ has a unique representation in terms of its residues modulo $q_i$, for $i=1, \dots, d$. Let $\psi: \mathbb{Z}\mod q \rightarrow GF(q_1) \times \dots \times GF(q_d)$ denote this map.
\begin{itemize}
\item Suppose that $(n,k_i)$ cyclic codes over $GF(q_i)$ exist for all $i = 1, \dots, d$. Each individual code is denoted $\calC^{i}$.
\item Let $\calC$ denote the code over $\mathbb{Z} \mod q$. Let $\bfc^{(i)} \in \calC^{i}$ for $i = 1, \dots, d$. The codeword $\bfc \in \calC$ is obtained as follows. The $j$-th component of $\bfc$, $\bfc_j = \psi^{-1}(\bfc^{(1)}_j, \dots, \bfc^{(d)}_j)$
\end{itemize}

Therefore, there are $q_1^{k_1}q_2^{k_2}\cdots q_d^{k_d}$ codewords in $\calC$. It is also evident that $\calC$ is cyclic.
As discussed in Section \ref{sec:construction}, we form the matrix $\bfT$ for the codewords in $\calC$. It turns out that using $\bfT$ and the technique discussed in Section \ref{sec:construction}, we can obtain a resolvable design. Furthermore, the gain of the system in the delivery phase can be shown to be $k_{min}=\min\{k_1,k_2,\cdots,k_d\}$. We discuss these points in detail in the Appendix (Section \ref{sec:blake_codes}).


\section{Discussion and Comparison with Existing Schemes}
\label{sec:compare}
\subsection{Discussion}



When the number of users is $K=nq$ and the cache fraction is $\frac{M}{N}=\frac{1}{q}$, we have shown in Theorem \ref{them:main} that the gain $g=k+1$ and $F_s=q^kz$. Therefore, both the gain and the subpacketization level increase with larger $k$. Thus, for our approach given a subpacketization budget $F'_s$, the highest coded gain that can be obtained is denoted by $g_{max}=k_{max}+1$ where $k_{max}$ is the largest integer such that $q^{k_{max}}z\le F'_s$ and there exists a $(n,k_{max})$ linear block code that satisfies the CCP.

For determining $k_{max}$, we have to characterize the collection of values of $k$ such that there exists a $(n,k)$ linear code satisfies the CCP over $GF(q)$ or $\mathbb{Z} \mod q$. We use our proposed constructions (MDS code, Claim \ref{claim:cyclic_MDS}, Claim \ref{claim:Gconstruct1}, Claim \ref{claim:Gconstruct2}, Claim \ref{claim:matrixextend}, Claim \ref{claim:SPC_CCPZmodq}, Claim \ref{claim:z_2_Zmodq}) for this purpose. We call this collection $\calC(n,q)$ and generate it in Algorithm \ref{Alg:CandidateK}. We note here that it is entirely possible that there are other linear block codes that fit the appropriate parameters and are outside the scope of our constructions. Thus, the list may not be exhaustive. In addition, we note that we only check for the $(k,k+1)$-CCP. Working with the $(k,\alpha)$-CCP where $\alpha \leq k$ can provide more operating points. 



%
%
{
	\renewcommand{\arraystretch}{1.5}
	\begin{table*}[t]
		\begin{center}	
		\begin{tabular}{|c|c|c|c|c|c|}
			\hline
			\hline
			$k$ & $n'$ & $z$ & $\alpha$ & \text{Construction} & \text{Notes}\\
			\hline
			\hline 
			11  & 12  & 1  & 1  & (12, 11)\text{~SPC code} & $k+1=n'$\\
			\hline
			10  & 12  & 11  & 12   & -&-\\
			\hline
			9  & 12  & 5  & 6    & \text{Claim \ref{claim:Gconstruct2}} & $\alpha = z+1$ and $q\ge z$\\
			\hline
			8  & 12  & 3  & 4     & \text{Claim \ref{claim:Gconstruct2}} &  $\alpha = z+1$ and $q\ge z$\\
		    \hline
			7  & 12  & 2  & 3    & \text{Claim \ref{claim:Gconstruct2}} & $\alpha = z+1$ and $q\ge z$\\
		    \hline
			6  & 12  & 7  & 12     & \text{Claim \ref{claim:cyclic_MDS}} &\text{Generator polynomial is} $X^6+X^5+3X^4+3X^3+X^2+4X+3$\\
		    \hline
			5  & 6  & 1  & 1     & \text{(6,5) SPC code and Claim \ref{claim:matrixextend}} &\text{Extend (6,5) SPC code to (12,5) code}\\
		    \hline
			4  & 7  & 5  & 7     & \text{Claim \ref{claim:cyclic_MDS}} &\text{Generator polynomial is} $X^8+X^7+4X^6+3X^5+2X^3+X^2+4X+4$\\
		    \hline
			3  & 4  & 1  & 1     & \text{(4,3) SPC code and Claim \ref{claim:matrixextend}} &\text{Extend (4,3) SPC code to (12,3) code} \\
		    \hline
			2  & 3  & 1  & 1     &  \text{(3,2) SPC code and Claim \ref{claim:matrixextend}} &\text{Extend (3,2) SPC code to (12,2) code}\\
			\hline
			1  & 2  & 1  & 1     &  \text{(2,1) SPC code and Claim \ref{claim:matrixextend}} &\text{Extend (2,1) SPC code to (12,1) code}\\
			\hline
			\hline
		\end{tabular}
    \end{center}
    \caption{\label{Table:Example} List of $k$ values for Example \ref{eg:discuss_comp_1}. The values of $n', \alpha$ and $z$ are obtained by following Algorithm \ref{Alg:CandidateK}.}
	\end{table*}
}
\begin{example}
\label{eg:discuss_comp_1}
	Consider a caching system with $K=nq=12\times 5=60$ users and cache fraction $\frac{M}{N}=\frac{1}{5}$. Suppose that the subpacketization budget is $1.5\times 10^6$. By checking all $k<n$ we can construct $\calC(n,q)$ (see Table \ref{Table:Example}). 
As a result, $\calC(n,q)=\{1,2,3,4,5,6,7,8,9,11\}$. Then $k_{max}=8$, $F_s\approx 1.17\times 10^6$ and the maximal coded gain we can achieve is $g_{max}=9$. By contrast, the scheme in \cite{maddahN14} can achieve coded gain $g=\frac{KM}{N}+1=13$ but requires subpacketization level $F_s=\binom{K}{\frac{KM}{N}}\approx 1.4\times 10^{12}$.

We can achieve almost the same rate by performing memory-sharing by using the scheme of \cite{maddahN14} in this example. In particular, we divide each file of size $\Omega$ into two smaller subfiles $W_n^1$ and $W_n^2$, where the size of $W_n^1$, $|W_n^1|=\frac{9}{10} \Omega$ and the size of $W_n^2$, $|W_n^2|=\frac{1}{10}\Omega$. The scheme of \cite{maddahN14} is then applied separately on $W_n^1$ and $W_n^2$ with $\frac{M_1}{N_1}=\frac{2}{15}$ (corresponding to $W_n^1$) and $\frac{M_2}{N_2}=\frac{13}{15}$ (corresponding to $W_n^2$). Thus, the overall cache fraction is $0.9 \times \frac{2}{15} + 0.1 \times \frac{13}{15} \approx \frac{1}{5}$. The overall coded gain of this scheme is $g \approx 9$. However, the subpacketization level is $F_s^{MN} = \binom{K}{KM_1/N_1} + \binom{K}{KM_2/N_2}\approx 5\times 10^9$, which is much greater than the subpacketization budget.
\end{example}

In Fig. \ref{fig:rate_comps}, we present another comparison for system parameters $K=64$ and different values of $M/N$. The scheme of \cite{maddahN14} works for all $M/N$ such that $KM/N$ is an integer. In Fig. \ref{fig:rate_comps}, our plots have markers corresponding to $M/N$ values that our scheme achieves. For ease of presentation, both the rate (left $y$-axis) and the logarithm of the subpacketization level (right $y$-axis) are shown on the same plot. We present results corresponding to two of our construction techniques: (i) the SPC code and (ii) a smaller SPC code coupled with Claim \ref{claim:matrixextend}. It can be seen that our subpacketization levels are several orders of magnitude smaller with only a small increase in the rate.

\begin{figure}[!t]
		\centering
		\includegraphics[scale=0.35]{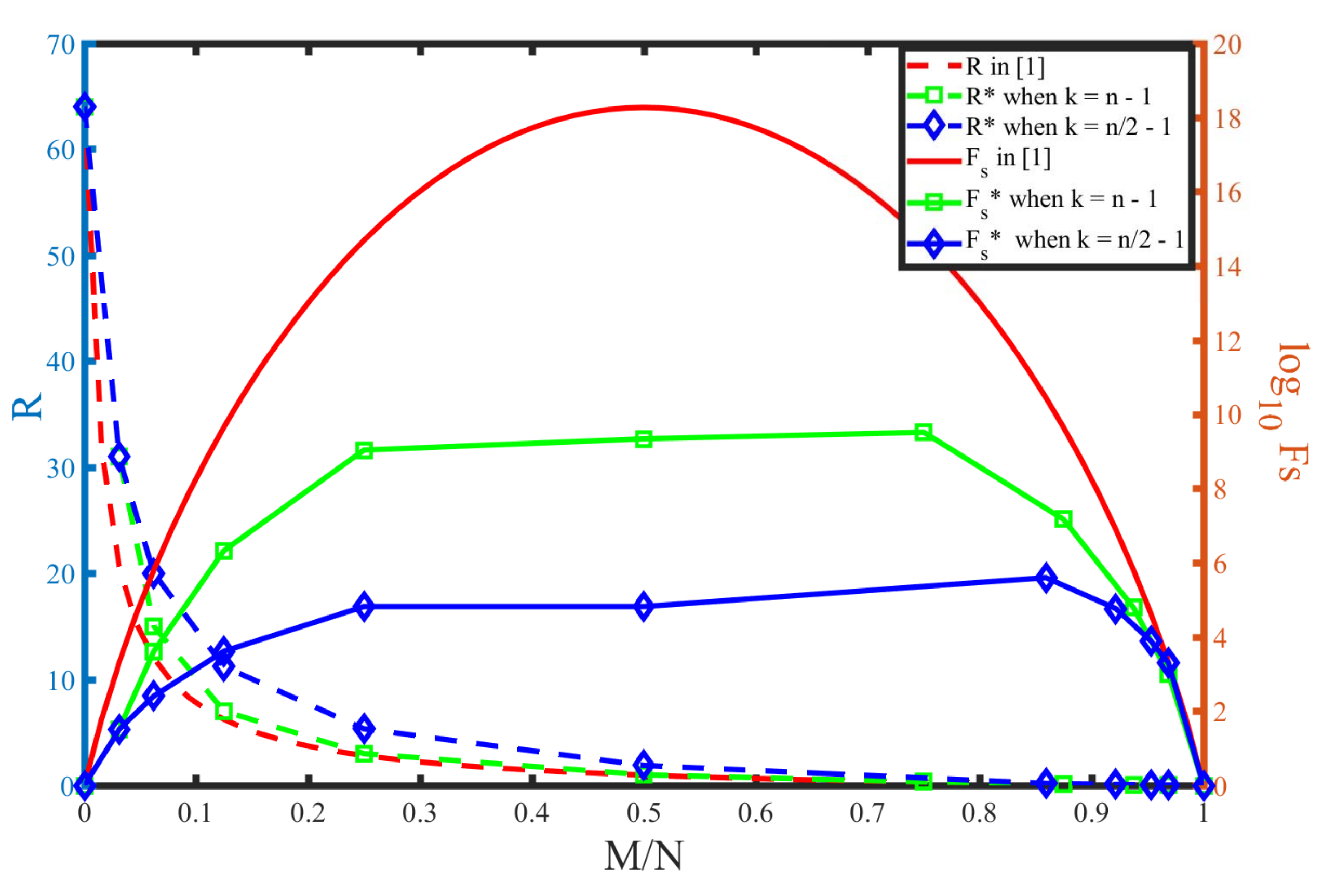}
		\caption{A comparison of rate and subpacketization level vs. $M/N$ for a system with $K=64$ users. The left $y$-axis shows the rate and the right $y$-axis shows the logarithm of the subpacketization level. The green and the blue curves correspond to two of our proposed constructions. Note that our schemes allow for multiple orders of magnitude reduction in subpacketization level and the expense of a small increase in coded caching rate.}
		\label{fig:rate_comps}
\end{figure}

An in-depth comparison for general parameters is discussed next. In the discussion below, we shall use the superscript $*$ to refer to the rates and subpacketization levels of our proposed scheme.

\begin{algorithm}[t]
	\caption{$\calC(n,q)$ Construction Algorithm}
	\label{Alg:CandidateK}
	\SetKwInOut{Input}{Input}
	\SetKwInOut{Output}{Output}
	\Input{$n$, $q$, $\calC(n,q)=\emptyset$}
	\If{$q$ is a prime power}
	{
		\For{$k=1:(n-1)$}
		{
			$n'\leftarrow (n)_{k+1}+k+1$\;
			$z\leftarrow \frac{k+1}{\gcd(n',k+1)}$\;
            $\alpha \leftarrow \frac{n'}{\gcd(n',k+1)}$\;
            \eIf{there exists a $(n' + i(k+1),k)$ cyclic code which satisfies the condition in Claim \ref{claim:cyclic_MDS} for some $i$ such that $n' + i(k+1) \leq n$ }{$\calC(n,q)\leftarrow k$. Corresponding codes are constructed by using Claim \ref{claim:matrixextend}.}{
			\eIf{$z\le 2$}{$\calC(n,q)\leftarrow k$. Corresponding codes are constructed by SPC code and Claim \ref{claim:matrixextend} when $z=1$ or Claim \ref{claim:Gconstruct2} and Claim \ref{claim:matrixextend} when $z=2$.}{
			\eIf{$q+1 \geq n'$}{$\calC(n,q)\leftarrow k$. Corresponding codes are constructed by MDS code and Claim \ref{claim:matrixextend}.}
		                   { \If{$\alpha=z+1$ and $q\ge z$} {$\calC(n,q)\leftarrow k$. Corresponding codes are constructed by Claim \ref{claim:Gconstruct2} and Claim \ref{claim:matrixextend}.}
		                     \If{$\alpha>z+1$ and $q > \alpha$} {$\calC(n,q)\leftarrow k$. Corresponding codes are constructed by Claim \ref{claim:Gconstruct1} and Claim \ref{claim:matrixextend}.}
                           }
           }
        }
        }
    }
    \If{$q$ is not a prime power}
    {
    	\For{$k=1:(n-1)$}
    	{
    		\If{$z\le 2$}
    		{$\calC(n,q)\leftarrow k$. Corresponding codes are constructed by Claim \ref{claim:SPC_CCPZmodq} when $z=1$ and Claim \ref{claim:matrixextend} or Claim \ref{claim:z_2_Zmodq} and Claim \ref{claim:matrixextend}  when $z=2$.
    		}
    	}
    }
	\Output{$\calC(n,q)$}
\end{algorithm}

\subsection{Comparison with memory-sharing within the scheme of \cite{maddahN14}}

Suppose that for given $K$, $M$ and $N$, a given rate $R$ can be achieved  by the memory sharing of the scheme in \cite{maddahN14} between the corner points $(M_1, R_1), (M_2,R_2),\cdots, (M_d,R_d)$ where $M_i=\frac{t_iN}{K}$ for some integer $t_i$. Then $R_i=\frac{K(1-\frac{M_i}{N})}{1+\frac{KM_i}{N}}$, $R=\sum_{i=1}^d \lambda_i R_i$, $M/N = \sum_{i=1}^d \lambda_i \frac{M_i}{N}$ and $\sum_{i=1}^d \lambda_i=1$. The subpacketization level is $F_s^{MS}=\sum_{i=1}^d\binom{K}{\frac{KM_i}{N}}$. In addition, we note that the function $h(x) = K(1 - x)/(1 + Kx)$ is convex in the parameter $0 \leq x \leq 1$. This can be verified by a simple second derivative calculation.

We first argue that $F_s^{MS}$ is lower bounded by $\binom{K}{\frac{KM'}{N}}$, where $M'$ is obtained as follows. For a given $M/N$, we first determine $\lambda$ and $\frac{M^*}{N}$ that satisfy the following equations.
\begin{align}
R &=\lambda\frac{K(1-\frac{M^*}{N})}{1+\frac{KM^*}{N}}+(1-\lambda)\frac{K\frac{M^*}{N}}{1+K(1-\frac{M^*}{N})}, \text{~and~} \label{eq:memoryR}\\
\frac{M}{N} &=\lambda \frac{M^*}{N}+(1-\lambda)\bigg{(} 1-\frac{M^*}{N} \bigg{)}. \label{eq:memoryM}
\end{align}
Here, $\frac{M^*}{N}\le \frac{1}{2}$, and $M'=\frac{t'N}{K}$, where $t'$ is the least integer such that $M'\ge M^*$.

To see this, consider the following argument. Suppose that the above statement is not true. Then, there exists a scheme that operates via memory sharing between points $(M_1,R_1), \cdots, (M_d,R_d)$ such that $F_s<\binom{K}{K\frac{M'}{N}}$. Note that $\binom{K}{\frac{K M_1}{N}}<\binom{K}{\frac{K M_2}{N}}$ if $\frac{M_1}{N}<\frac{M_2}{N}\le \frac{1}{2}$ or $\frac{M_1}{N}>\frac{M_2}{N}\ge \frac{1}{2}$. By the convexity of $h(\cdot)$, we can conclude that $(M, R)$ is not in the convex hull of the corner points $(M_1,R_1), \cdots, (M_d,R_d)$. This is a contradiction.

Next, we compare this lower bound on $F_s^{MS}$ to the subpacketization level of our proposed scheme. In principle, we can solve the system of equations (\ref{eq:memoryR}) and (\ref{eq:memoryM}) for $R = \frac{n(q-1)}{k+1}$ and $\frac{M}{N} = \frac{1}{q}$ and obtain the appropriate $\lambda$ and $M^*$ values\footnote{Similar results can be obtained for $\frac{M}{N}=1-\frac{k+1}{nq}$}. Unfortunately, doing this analytically becomes quite messy and does not yield much intuition. Instead, we illustrate the reduction in subpacketization level by numerical comparisons.

\begin{example}
	Consider a $(9,5)$ linear block code over $GF(2)$ with generator matrix specified below.
	\begin{align*}
	\bfG=
	\begin{bmatrix}
        1&0&0&0&0&1&1&0&0\\
        0&1&0&0&0&1&0&1&0\\
        0&0&1&0&0&1&0&0&1\\
        0&0&0&1&0&1&1&1&0\\
        0&0&0&0&1&1&0&1&1
    \end{bmatrix}.	
	\end{align*}
    It can be checked that $\bfG$ satisfies the $(5,6)$-CCP. Thus, it corresponds to a coded caching system with $K=9\times 2=18$ users. Our scheme achieves the point $\frac{M_1}{N}=\frac{1}{2}$, $R_1=\frac{3}{2}$, $F_{s,1}^*=64$ and $\frac{M_2}{N}=\frac{2}{3}$, $R_2=\frac{2}{3}$, $F_{s,2}^*=96$.

    On the other hand for $\frac{M_1}{N}=\frac{1}{2}$, $R_1=\frac{3}{2}$, by numerically solving  (\ref{eq:memoryR}) and (\ref{eq:memoryM}) we obtain $\frac{M_1^*}{N}\approx 0.227$ and therefore $\frac{M_1'}{N}=\frac{5}{18}$. Then $F_{s,1}^{MS}\ge\binom{18}{5}=8568$, which is much higher than $F_{s,1}^*=64$. A similar calculation shows that $\frac{M_2^*}{N}\approx \frac{1}{4}$ and therefore $\frac{M_2'}{N}=\frac{5}{18}$. Thus  $F_{s,2}^{MS}$ is also at least as large as $8568$, which is still much higher than $F_{s,2}^*=96$.
\end{example}

The next set of comparisons are with other proposed schemes in the literature. We note here that several of these are restrictive in the parameters that they allow.
\subsection{Comparison with \cite{maddahN14}, \cite{yan_et_al17},\cite{yan2016placement}, \cite{shangguan2016centralized} and \cite{shanmugam2017coded}}

For comparison with \cite{maddahN14}, denote  $R^{MN}$  and $F_s^{MN}$ be the rate and the subpacketization level of the scheme of \cite{maddahN14}, respectively. For the rate comparison, we note that
\begin{align*}
 \frac{R^*}{R^{MN}}=\frac{1+n}{1+k},&\text{~~for~} \frac{M}{N}=\frac{1}{q}\\
\frac{R^*}{R^{MN}}=\frac{nq-k}{nq-n},&\text{~~for~} \frac{M}{N}=1-\frac{1+k}{nq},
\end{align*}

For the comparison of  subpacketization level we have the following results.
\begin{claim} When $K = nq$, the following results hold.
	\label{claim:compare_MN}
	\begin{itemize}
		\item If  $\frac{M}{N}=\frac{1}{q}$, we have
		\begin{align}
		\lim_{n \rightarrow \infty} \frac{1}{K}\log_2\frac{F_s^{MN}}{F_s^*} = H_2\bigg{(}\frac{1}{q}\bigg{)} -\frac{\eta}{q} \log_2 q. \label{eq:scaling_gain}
		\end{align}
		\item If  $\frac{M}{N}=1-\frac{k+1}{nq}$, we have 
		\begin{align}
		\lim_{n \rightarrow \infty}\frac{1}{K}\log_2\frac{F_s^{MN}}{F_s^*}= H_2\bigg{(}\frac{\eta}{q}\bigg{)}-\frac{\eta}{q} \log_2 q.
		\end{align}
	\end{itemize}
	In the above expressions, $0 < \eta =k/n \le 1$ and $H_2(\cdot)$ represents the binary entropy function. 
\end{claim}
\begin{proof}
	Both results are simple consequences of approximating $\binom{K}{Kp}\approx 2^{KH_2(p)}$ \cite{graham1994concrete}. The derivations can be found in the Appendix.
\end{proof}
It is not too hard to see that $F_s^*$ is exponentially lower than $F_s^{MN}$. Thus, our rate is higher, but the subpacketization level is exponentially lower.
Thus, the gain in the scaling exponent of with respect to the scheme of \cite{maddahN14} depends on the choice of $R$ and the value of $M/N$. In Fig. \ref{fig:scaling_exponent_gain} we plot this value of different values of $R$ and $q$. The plot assumes that codes satisfying the CCP can be found for these rates and corresponds to the gain in eq. (\ref{eq:scaling_gain}).
\begin{figure}[t]
		\centering
		\includegraphics[scale=0.5]{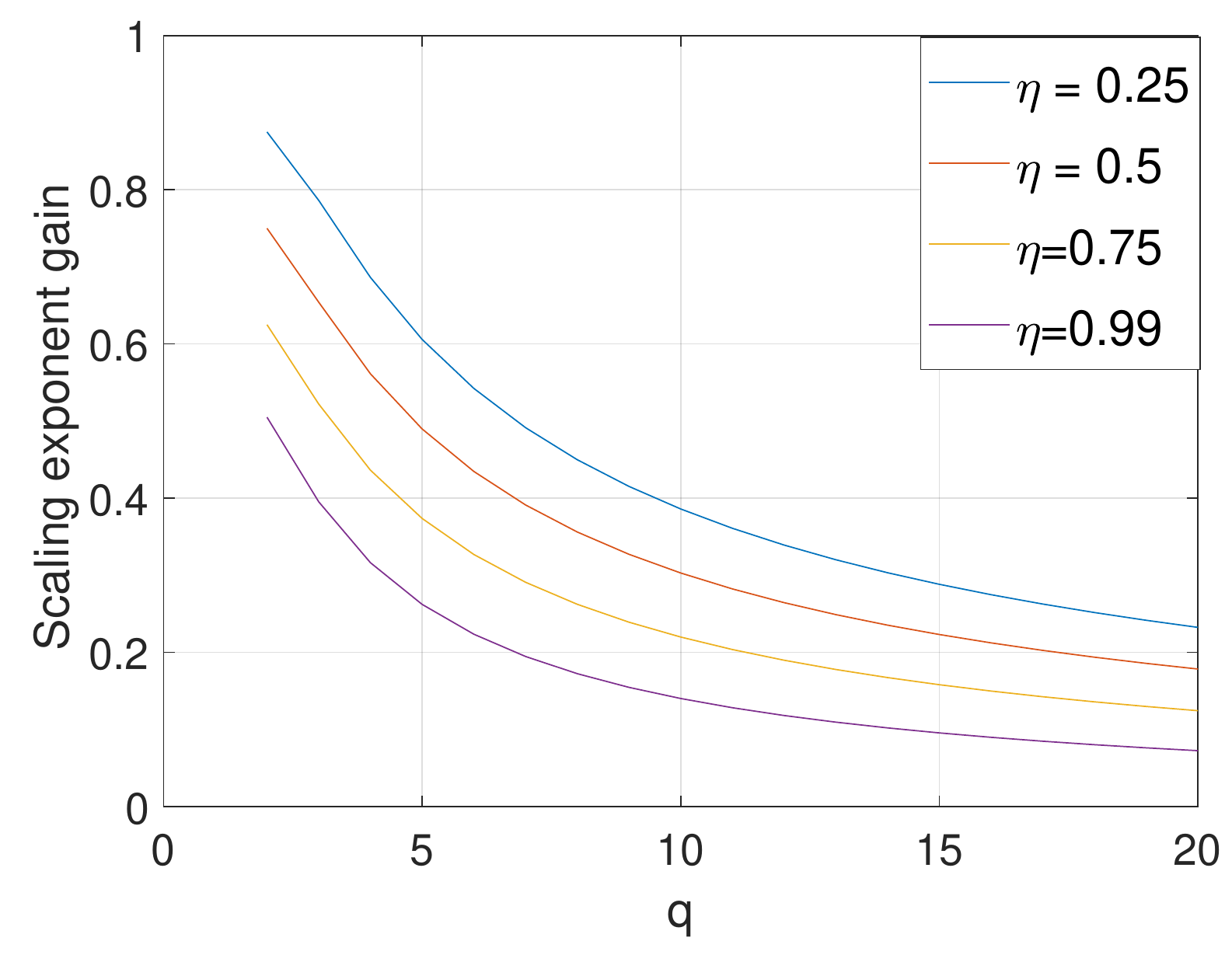}
		\caption{The plot shows the gain in the scaling exponent obtained using our techniques for different value of $M/N = 1/q$. Each curve corresponds to a choice of $\eta = k/n$.}
		\label{fig:scaling_exponent_gain}

\end{figure}

%

In \cite{yan_et_al17} a scheme for the case when $M/N = 1/q$ or $M/N = 1 - 1/q$ with subpacketization level exponentially smaller with respect to \cite{maddahN14} was presented. This result can be recovered a special case of
our work (Theorem \ref{them:main}) when the linear block code is chosen as a single parity check code over $\mathbb{Z} \mod q$. In this specific case, $q$ does not need to be a prime power. Thus, our results subsume the results of \cite{yan_et_al17}.

In a more recent preprint, reference \cite{yan2016placement}, proposed a caching system with $K =\binom{m}{a}$, $F_s=\binom{m}{b}$ and $M/N = 1 -\binom{a}{\lambda}\binom{m-a}{b-\beta}/ \binom{m}{b}$. The corresponding rate is 
$$R = \frac{\binom{m}{a+b-2\beta}}{\binom{m}{b}}\min \bigg{\{} \binom{m-(a+b-2\beta)}{\lambda},\binom{a+b-2\beta}{a-\beta} \bigg{\}},$$
where $m,a,b,\beta$ are positive integers and $0<a<m$, $0<b<m$, $0\le \beta \le \min\{a,b\}$. While a precise comparison is somewhat hard, we can compare the schemes for certain parameter choices, that were also considered in \cite{yan2016placement}.

Let $a=2$, $\beta=1$, $m=2b$. This corresponds to a coded caching system with $K=b(2b-1)\approx 2b^2$, $\frac{M}{N}=\frac{b-1}{2b-1}\approx \frac{1}{2}$, $F_s=\binom{2b}{b}\approx 2^{2b}$, $R=b$. For comparison with our scheme we keep the transmission rates of both schemes roughly the same and let $n=b^2$, $q=2$, $k=b-1$. We assume that the corresponding linear block code exists. Then $F^*_s\approx 2^b$, which is better than $F_s$.

On the other hand if we let $\beta=0$, $a=2$, $m=2qb$, we obtain a coded caching system with
$K=\frac{m(m-1)}{2}$, $\frac{M}{N}\approx \frac{1}{q}$, $F_s=\binom{m}{\frac{m}{2q}}\approx (2q)^{\frac{m}{2q}}$, $R^{YAN}=(2q-1)^2$.
For keeping the rates the same, we let $n=\frac{m(m-1)}{2q}$, $k= \frac{m(m-1)}{4q(2q-1)}-1$ so that  $F^*_s\approx q^{\frac{m(m-1)}{4q(2q-1)}}\approx q^{\frac{m^2}{8q^2}}$. In this regime, the subpacketization level of \cite{yan2016placement} will typically be lower.

The work of \cite{shangguan2016centralized} proposed caching schemes with parameters (i) $K=\binom{m}{a}$, $\frac{M}{N}=1-\frac{\binom{m-a}{b}}{\binom{m}{b}}$, $F_s=\binom{m}{b}$ and $R=\frac{\binom{m}{a+b}}{\binom{m}{b}}$, where $a,b,m$ are positive integers and $a+b\le m$ and (ii) $K=\binom{m}{t}q^t$, $\frac{M}{N}=1-\frac{1}{q^t}$, $F_s=q^m(q-1)^t$ and $R=\frac{1}{(q-1)^t}$, where $q,t,m$ are positive integers.

Their scheme (i), is the special case of scheme in \cite{yan2016placement} when $\beta=0$. For the second scheme, if we let $t=2$, \cite{shangguan2016centralized} shows that $R\approx R^*$, $F_s\approx q^{\sqrt{\frac{K}{2q}}}(\sqrt{q}-1)^2$ and $F_s^*\approx (q-1)q^{\frac{K}{q}-1}$, which means $F_s$ is again better than $F_s^*$. We emphasize here that these results require somewhat restrictive parameter settings.

Finally, we consider the work of \cite{shanmugam2017coded}. In their work, they leveraged the results of \cite{alon2012nearly} to arrive at coded caching schemes where the subpacketization is linear in $K$. Specifically, they show that for any constant $M/N$, there exists a scheme with rate $K^\delta$, where $\delta > 0$ can be chosen arbitrarily small by choosing $K$ large enough. From a theoretical perspective, this is a positive result that indicates that regimes where linear subpacketization scaling is possible. However, these results are only valid when the value of $K$ is very large. Specifically, $K = C^n$ and the result is asymptotic in the parameter $n$. For these parameter ranges, the result of \cite{shanmugam2017coded} will clearly be better as compared to our work.

\section{Conclusions and Future Work}
\label{sec:conclusion}

In this work we have demonstrated a link between specific classes of linear block codes and the subpacketization problem in coded caching. Crucial to our approach is the consecutive column property which enforces that certain consecutive column sets of the corresponding generator matrices are full-rank. We present several constructions of such matrices that cover a large range of problem parameters. Leveraging this approach allows us to construct families of coded caching schemes where the subpacketization level is exponentially smaller compared to the approach of \cite{maddahN14}.

There are several opportunities for future work. Even though our subpacketization level is significantly lower than \cite{maddahN14}, it still scales exponentially with the number of users. Of course, the rate of growth with the number of users is much smaller. There have been some recent results on coded caching schemes that demonstrate the existence of schemes where the subpacketization scales sub-exponentially in the number of users. It would be interesting to investigate whether some of these ideas can be leveraged to obtain schemes that work for practical systems with tens or hundreds of users.



\appendix
\subsection{Resolvable design over $Z\text{~mod~}q$}
\begin{lemma}
	\label{lemma:resol_ring}
	A $(n,k)$ linear block code over $\mathbb{Z} \text{~mod~}q$ with generator matrix $\bfG=[g_{ab}]$ can construct a resolvable block design by the procedure in Section \ref{sec:construction} if $\gcd(q,g_{0b},g_{1b},\cdots,g_{(k-1)b})=1$ for $0\le b<n$.
\end{lemma}
\begin{proof}
    Assume $q=q_1\times q_2\times \cdots \times q_d$ where $q_i$, $1\le i\le d$ is a prime or a prime power. If the $\gcd(q,g_{0b}, g_{1b}, \cdots, g_{(k-1)b})=1$, then it is evident that $\gcd(q_i,g_{0b}, g_{1b},\cdots, g_{(k-1)b})=1$ for $1\le i\le d$. As $q_i$ is either a prime or a prime power, it follows that there exists a $g_{a^*b}$ which is relatively prime to $q_i$, i.e., $g_{a^*b}$ is a unit in the ring $\mathbb{Z} \mod q_i$.

    Note that for $\Delta = [\Delta_0 ~\Delta_1~ \dots~\Delta_{n-1}]= \bfu \bfG$, we have
	\begin{equation}
	\label{eq:resolring}
	\Delta_b = \sum_{a=0}^{k-1} \bfu_ag_{ab},
	\end{equation}
	where $\bfu=[\bfu_0,\cdots, \bfu_{k-1}]$. We consider eq. (\ref{eq:resolring}) over the ring $\mathbb{Z} \mod q_i$ and rewrite eq. (\ref{eq:resolring}) as
	\begin{align*}
	\Delta_b-\bfu_{a^*}g_{a^*b}=\sum_{a\neq a^*}\bfu_ag_{ab},
	\end{align*}
	For arbitrary $\bfu_a$, $a\neq a^*$, this equation has a unique solution for $\bfu_{a^*}$ since $g_{a^*b}$ is a unit in $\mathbb{Z} \mod q_i$. This implies that there are $q_i^{k-1}$ distinct solutions for (\ref{eq:resolring}) over $\mathbb{Z} \mod q_i$. Using the Chinese remainder theorem, eq. (\ref{eq:resolring}) has $q_1^{k-1}\times q_2^{k-1}\times \cdots q_d^{k-1}=q^{k-1}$ solutions over $\mathbb{Z} \mod q$ and the result follows.
\end{proof}

\begin{remark}
	From Lemma \ref{lemma:resol_ring}, it can be easily verified that a linear block code over $\mathbb{Z}\text{~mod~}q$ can construct a resolvable block design if one of the following conditions for each column $\bfg_i$ of the generator matrix is satisfied.
	\begin{itemize}
		\item At least one non-zero entry of $\bfg_i$ is a unit in $\mathbb{Z} \mod q$, or
		\item all non-zero entries in $\bfg_i$ are zero divisors but their greatest common divisor is 1.
	\end{itemize}	
	For the SPC code over $\mathbb{Z}\text{~mod~}q$, all the non-zero entries in the generator matrix are $1$, which is an unit. Therefore, the construction always results in a resolvable design.
\end{remark}
\subsection*{Proof of Lemma \ref{lemma:delivery}}
First, we show that the proposed delivery scheme allows each user's demand to be satisfied. Note that Claim \ref{claim:signalgeneration} shows that each user in a parallel class that belongs to the recovery set $\calS_a$ recovers all missing subfiles with a specified superscript from it. Thus, we only need to show that if signals are generated (according to Claim \ref{claim:signalgeneration}) for each recovery set, we are done. This is equivalent to showing that the bipartite recovery set graph is such that each parallel class has degree $z$ and multiple edges between nodes are disallowed.


Towards this end, consider the parallel class and we claim that there exist exactly $z$ solutions $(a_\alpha, b_\alpha)$ for integer values of $\alpha=1, \dots, z$ to the equation
\begin{align}
a_\alpha (k+1) + b_\alpha &= j + n(\alpha-1) \label{eq:congruence}
\end{align}
such that $a_{\alpha_1} \neq a_{\alpha_2}$ for $\alpha_1 \neq \alpha_2$ and $j < n$.
The existence of the solution for each equation above follows from the division algorithm. Note that $a_\alpha < nz/(k+1)$ as the $RHS < nz$. Furthermore, note that for $1 \leq \alpha_1 \leq z$ and $1 \leq \alpha_2 \leq z$, we cannot have solutions to eq. (\ref{eq:congruence}) such that $a_{\alpha_1} = a_{\alpha_2}$ as this would imply that $ |b_{\alpha_1} - b_{\alpha_2}| \geq n$ which is a contradiction.
This shows that each parallel class $P_j$ participates in at least $z$ different recovery sets.

The following facts follow easily from the construction of the recovery sets. The degree of each recovery set in the bipartite graph is $k+1$ and there are $\frac{nz}{k+1}$ of them; multiple edges between a recovery set and a parallel class are disallowed. Therefore, the total number of edges in the bipartite graph is $nz$. As each parallel class participates in at least $z$ recovery sets, by this argument, it participates in exactly $z$ recovery sets.
.

Finally, we calculate the rate of the delivery phase. In total, the server transmits $q^k(q-1)\frac{zn}{k+1}$ equations, where the symbol transmitted has the size of a subfile. Thus, the rate is
\begin{align*}
R&=q^k(q-1)\frac{zn}{k+1}\frac{1}{q^kz}\\
&=\frac{(q-1)n}{k+1}.
\end{align*}
\endproof

%

\subsection*{Proof of Claim \ref{claim:cyclic_MDS}}
The matrix  $\bfG_\calS$ is shown below.
\begin{align*}
\begin{bmatrix}
g_0&g_1&\cdot&\cdot&g_{\lceil\frac{k}{2}\rceil-1}&0&\cdot&\cdot&0&0\\
0&g_0&g_1&\cdot&g_{\lceil\frac{k}{2}\rceil-2}&0&\cdot&\cdot&0&0\\
\vdots&&&&&&&&&\vdots\\
0&\cdot&0&g_0&g_1&0&\cdot&\cdot&0&0\\
0&\cdot&\cdot&0&g_0&g_{n-k}&0&\cdot&\cdot&0\\
0&\cdot&\cdot&0&0&g_{n-k-1}&g_{n-k}&0&\cdot&0\\
\vdots&&&&&&&&&\vdots\\
0&\cdot&\cdot&0&0&g_{n-k-\lfloor \frac{k}{2}\rfloor}&\cdot&\cdot&\cdot&g_{n-k}
\end{bmatrix}
\end{align*}
In the above expression and the subsequent discussion if $i$ is such that $i < 0$, we set $g_i=0$.

By Claim \ref{claim:cyclic_k1}, a cyclic code with generator matrix $\bfG$ satisfies the CCP if all submatrices \begin{align*}
\bfG_{\calS\setminus (a+j)_n}=&[\bfg_{(a)_n},\bfg_{(a+1)_n},\cdots,\bfg_{(a+j-1)_n},\\
&~\bfg_{(a+j+1)_n},\cdots,\bfg_{(a+k)_n}],
\end{align*} where $a=n-\lfloor \frac{k}{2}\rfloor-1$, $0\le j\le k$, have full rank. 
In what follows, we argue that this is true.
Note that in the generator matrix of cyclic code, any $k$ consecutive columns are linearly independent \cite{lincostello}. Therefore for $j=0$ and $k$, $\bfG_{\calS\setminus (a+j)_n}$ has full rank, without needing the conditions of Claim \ref{claim:cyclic_MDS}.
For $0<j\le \lfloor \frac{k}{2}\rfloor$, $\bfG_{\calS\setminus (a+j)_n}$ is as eq. (\ref{eq:G_cylic_1}).
\begin{table*}
	\begin{align}
	\label{eq:G_cylic_1}
	\begin{bmatrix}
	g_0&g_1&\cdot&g_{\lceil\frac{k}{2}\rceil-1}&0&\cdot&\cdot&\cdot&\cdot&0\\
	0  &g_0&\cdot&g_{\lceil\frac{k}{2}\rceil-2}&0&\cdot&\cdot&\cdot&\cdot&0\\
	\vdots&&&&&&&&&\vdots\\
	0&\cdot&g_0&g_1&0&\cdot&\cdot&\cdot&\cdot&0\\
	0&\cdot&0&g_0&g_{n-k}&0&\cdot&\cdot&\cdot&0\\
	0&\cdot&\cdot&0&g_{n-k-1}&g_{n-k}&0&\cdot&\cdot&0\\
	\vdots&&&&&&&&&\vdots\\
	0&\cdot&\cdot&0&g_{n-k-j+1}&\cdot&g_{n-k}&0&\cdot&0\\
	0&\cdot&\cdot&0&g_{n-k-j}&\cdot&g_{n-k-1}&0&\cdot&0\\
	0&\cdot&\cdot&0&g_{n-k-j-1}&\cdot&g_{n-k-2}&g_{n-k}&\cdot&0\\
	\vdots&&&&&&&&&\vdots\\
	0&\cdot&\cdot&0&g_{n-k-\lfloor \frac{k}{2}\rfloor}&\cdot&g_{n-k-\lfloor \frac{k}{2}\rfloor+j-1}&g_{n-k-\lfloor \frac{k}{2}\rfloor+j+1}&\cdot&g_{n-k}
	\end{bmatrix}.
	\end{align}
\end{table*}

Rewriting $\bfG_{\calS\setminus (a+j)_n}$ in block form, we get
\begin{align*}
\bfG_{\calS\setminus (a+j)_n}=
\left[
\begin{array}{c|cc}
\mathbf A_j&\multicolumn{2}{c}{\mathbf B_j}\\  \hline
\multirow{2}*{\Huge$\mathbf 0$}&\mathbf C_j&\mathbf 0\\
&\mathbf D_j&\mathbf E_j
\end{array}
\right],
\end{align*}
where
$$\mathbf A_j=\begin{bmatrix}
g_0&g_1&\cdot&\cdot&g_{\lceil\frac{k}{2}\rceil-1}\\
0&g_0&g_1&\cdot&g_{\lceil\frac{k}{2}\rceil-2}\\
\vdots&&&&\vdots\\
0&\cdot&&0&g_0
\end{bmatrix},$$
$$\mathbf C_j=\begin{bmatrix}
g_{n-k-1}&g_{n-k}&0&\cdot&\cdot&0\\
g_{n-k-2}&g_{n-k-1}&g_{n-k}&0&\cdot&0\\
\vdots&&&&&\vdots\\
g_{n-k-j+1}&\cdot&\cdot&\cdot&\cdot&g_{n-k}\\
g_{n-k-j}&\cdot&\cdot&\cdot&\cdot&g_{n-k-1}
\end{bmatrix},$$
and
$$\mathbf E_j=\begin{bmatrix}
g_{n-k}&0&\cdot&\cdot&0\\
g_{n-k-1}&g_{n-k}&0&\cdot&0\\
\vdots&&&&\vdots\\
g_{n-k-\lfloor \frac{k}{2}\rfloor+j+1}&\cdot&\cdot&\cdot&g_{n-k}
\end{bmatrix}.$$
Matrices $\mathbf A_j$ and $\mathbf E_j$ have full rank as they are respectively upper triangular and lower triangular, with non-zero entries on the diagonal (as $g_0$ and $g_{n-k}$ are non-zero in a cyclic code). Therefore, $\bfG_{\calS\setminus (a+j)_n}$ has full rank if $\mathbf C_j$ has full rank.
For $\lfloor \frac{k}{2}\rfloor<j<k$, $\bfG_{\calS\setminus (a+j)_n}$ can be partitioned into a similar form and the result in Claim \ref{claim:cyclic_MDS} follows.
\endproof

\subsection*{Proof of Claim \ref{claim:Gconstruct1}}
We need to argue that all $k\times k$ submatrices of $\bfG_{\calS_{a}}$ where $0\le a< \alpha$ are full rank. In what follows we argue that all $k\times k$ submatrices of $\bfG_{\calS_0}$ are full rank. The proof for any $\bfG_{\calS_a}$ is similar. Note that $\bfG_{\calS_0}$ can be written compactly as follows by using Kronecker products.
\begin{align*}
	\bfG_{\calS_0}=
	\begin{bmatrix}
		\bfA\otimes \bfI_t \\
		\bfB\otimes \bfC_{(t-1)\times t}(1,1)
	\end{bmatrix},
\end{align*}
where 	
\begin{align*}
\bfA &=
\begin{bmatrix}
b_{00}&\cdots&b_{0(z-1)}\\
\vdots&&\vdots\\
b_{(z-2)0}&\cdots&b_{(z-2)(z-1)}
\end{bmatrix}
\end{align*}
and $\bfB =[b_{(z-1)0},\cdots,b_{(z-1)(z-1)}]$.

Next, we check the determinant of submatrices $\bfG_{\calS_0\setminus i}$ obtained by deleting $i$-th column of $\bfG_{\calS_0}$. W.l.o.g, we let $i=(z-1)t+j$ where $0 \leq j < t$. The block form of the resultant matrix $\bfG_{\calS_0\setminus i}$ can be expressed as
\begin{align*}	
	\bfG_{\calS_0\setminus i}&=
	\begin{bmatrix}
	\bfA'\otimes \bfI_t&\bfA''\otimes \Delta_1\\
	\bfB'\otimes \bfC_{(t-1)\times t}(1,1)&\bfB''\otimes \Delta_2
	\end{bmatrix},
\end{align*}
where $\bfA'$ and $\bfA''$ are the first $z-1$ columns and last column of $\bfA$ respectively. Likewise, $\bfB'$ and $\bfB''$ are the first $z-1$ components and last component of $\bfB$. The matrices $\Delta_1$ and $\Delta_2$ are obtained by deleting the $j$-th column of $\bfI_t$ and $\bfC_{(t-1)\times t}(1,1)$ respectively. Then, using the Schur determinant identity \cite{hornJ91}, we have
\begin{align*}
\det(\bfG_{\calS_0\setminus i})=&\det(\bfA'\otimes \bfI_t)\det(\bfB''\otimes \Delta_2-\\
&\bfB'\otimes \bfC_{(t-1)\times t}(1,1)\cdot (\bfA'\otimes \bfI_t)^{-1}\cdot \bfA''\otimes\Delta_1)\\
\overset{(1)}{=}&\det(\bfA'\otimes \bfI_t)\det(\bfB''\otimes \Delta_2-\\
&\bfB'\bfA'^{-1}\bfA''\otimes \bfC_{(t-1)\times t}(1,1)\Delta_1)\\
\overset{(2)}{=}&\det(\bfA'\otimes \bfI_t)\det((\bfB''-\bfB'\bfA'^{-1}\bfA'')\otimes \Delta_2),
\end{align*}
where $(1)$ holds by the properties of the Kronecker product \cite{hornJ91} and $(2)$ holds since $ \bfC_{(t-1)\times t}(1,1)\Delta_1=\Delta_2$. Next note that $\det(\Delta_2) \neq 0$. This is because $\Delta_2$ can be denoted as
	\begin{align*}
\Delta_2=
\left[
\begin{array}{c|c}
\mathbf A&\mathbf \bfzr\\  \hline
\bfzr&\mathbf B
\end{array}
\right],
\end{align*}
where $\bfA$ is a $j\times j$ upper-triangular matrix;
\begin{align*}
\bfA=\begin{bmatrix}
1&1&0&\cdots&0&0\\
0&1&1&\cdots&0&0\\
\vdots&&&&&\vdots\\
0&0&0&\cdots&1&1\\
0&0&0&\cdots&0&1
\end{bmatrix}
\end{align*}
and $\bfB$ is a $(t-1-j)\times (t-1-j)$ lower-triangular matrix;
\begin{align*}
\bfB=\begin{bmatrix}
1&0&\cdots&0&0&0\\
1&1&\cdots&0&0&0\\
\vdots&&&&&\vdots\\
0&0&\cdots&1&1&0\\
0&0&\cdots&0&1&1
\end{bmatrix}.
\end{align*}
Next, we define the matrix
\begin{align*}
\bfF= \begin{bmatrix}
b_{00}&b_{01}&\cdots&b_{0(z-1)}\\
b_{10}&b_{11}&\cdots&b_{1(z-1)}\\
\vdots&&&\vdots\\
b_{(z-1)0}&b_{(z-1)1}&\cdots&b_{(z-1)(z-1)}
\end{bmatrix}.
\end{align*}
Another application of the Schur determinant identity yields
\begin{align*}
\det(\bfB''-\bfB'\bfA'^{-1}\bfA'')&=\frac{\det(\bfF)}{\det(\bfA')}\\
&\neq 0,
\end{align*}
since $\det(\bfF)$ and $\det(\bfA')$ are both non-zero as their columns have the Vandermonde form. In $\bfF$, the columns correspond to distinct and non-zero elements from $GF(q)$; therefore, $q>z$. Note however, that the above discussion focused only $\bfG_{\calS_0}$. As the argument needs to apply for all $\bfG_{\calS_a}$ where $0 \leq a < \alpha$, we need $q > \alpha$.

\subsection*{Proof of Claim \ref{claim:Gconstruct2}}

Note that the matrix in eq. (\ref{eq:Gconstruct2}) is the generator matrix of $(n,k)$ linear block code over $GF(q)$ where $nz=(z+1)(k+1)$. Since $z$ and $z+1$ are coprime, $z$ is the least positive integer such that $k+1~|~nz$. To show $\bfG$ satisfies the CCP, we need to argue that all $k\times k$ submatrices of $\bfG_{\calS_a}$ where $0\le a\le z$ are full rank. It is easy to check that $\calS_a=\{0,\cdots, n-1\}\setminus \{t(z-a),t(z-a)+1,\cdots,t(z-a)+t-1\}$. We verify three types of matrix $\bfG_{\calS_a}$ as follows: I. $a=0$ II. $a=1$ III. $a>1$.
\begin{itemize}
	\item Type I
	
	When $a=0$, it is easy to verify that any $k\times k$ submatrix of $\bfG_{\calS_0}$ has full rank since $\bfG_{\calS_0}$ has the form $[\bfI_{k\times k}|\bfoe_k]$, which is the generator matrix of the SPC code.
	
	\item Type II
	
	When $a=1$, $\bfG_{\calS_1}$ has the form in eq. (\ref{eq:claim_6_type_2}),
		\begin{table*}[t]
	\begin{align}
		\bfG_{\calS_1}&=
		\begin{bmatrix}
		\begin{array}{cc}
		\underbrace{\begin{array}{cccc}
			\bfI_{t\times t}&\bfzr_{t\times t}&\cdots&\bfzr_{t\times t}\\
			\bfzr_{t\times t}&\bfI_{t\times t}&\cdots&\bfzr_{t\times t}\\
			\vdots&&&\\
			\bfzr_{t\times t}&\bfzr_{t\times t}&\cdots&\bfI_{t\times t}\\
			\bfzr_{(t-1)\times t}&\bfzr_{(t-1)\times t}&\cdots&\bfzr_{(t-1)\times t}
			\end{array}
			}_{\text{Case 1}}&
		\underbrace{\begin{array}{c}
			{b_1}\bfI_{t\times t}\\
			{b_2}\bfI_{t\times t}\\
			\vdots\\
			{b_{z-1}}\bfI_{t\times t}\\
			\bfC(c_1, c_2)_{(t-1)\times t}\\
			\end{array}
		}_{\text{Case 2}}
		\end{array}
				\end{bmatrix}
        \label{eq:claim_6_type_2}
	\end{align}
\end{table*}
	
	Case 1: Suppose that we delete any of first $(z-1)t$ columns in $\bfG_{\calS_1}$ (this set of columns is depicted by the underbrace in eq. (\ref{eq:claim_6_type_2})), say $i$-th column of $\bfG_{\calS_1}$, where $z_1t\le i< (z_1+1)t$ and $0\le z_1\le z-2$. Let $i_1=i-z_1t$, $i_2=(z_1+1)t-i-1$. The resultant matrix $\bfG_{\calS_1 \setminus i}$ can be expressed as follows.
	\begin{align*}
		\bfG_{\calS_1 \setminus i}=
		\left[
		\begin{array}{c|c}
			\mathbf A&\mathbf C\\  \hline
			\mathbf B&\mathbf D
		\end{array}
		\right],
	\end{align*}
	where
	\begin{align*}
		\mathbf A &= \bfI_{i\times i},\\
		\mathbf B &= \bfzr_{(k-i)\times i},\\
		\mathbf C &=
		\left[
		\begin{array}{ccc}
			\bfzr_{t\times (k-t-i)} &\multicolumn{2}{c}{{b_1}\bfI_{t\times t}} \\
			\bfzr_{t\times (k-t-i)} & \multicolumn{2}{c}{{b_2}\bfI_{t\times t}} \\
			\vdots &\multicolumn{2}{c}{\vdots}\\
			\bfzr_{t\times (k-t-i)} & \multicolumn{2}{c}{{b_{z_1}}\bfI_{t\times
					t}}\\
			\bfzr_{i_1\times (k-t-i)} &{b_{z_1+1}\bfI_{i_1\times
				i_1}}&\bfzr_{i_1\times (i_2+1)}
		\end{array}
		\right],
	\end{align*}
	and $\mathbf D$ has the form in eq. (\ref{eq:typeII_D}).
	\begin{table*}[t]
	\begin{align}
\mathbf D &=
\left[
\begin{array}{ccccccc}
\bfzr_{1\times i_2}  & \bfzr_{1\times t} & \cdots & \bfzr_{1\times t} & \bfzr_{1\times i_1} & b_{z_1+1} & \bfzr_{1\times i_2}\\
\bfI_{i_2\times i_2}  &\bfzr_{i_2\times t} &\cdots & \bfzr_{i_2\times t} &\multicolumn{2}{c}{\bfzr_{i_2\times (i_1+1)}} & {b_{z_1+1}}\bfI_{i_2\times
	i_2}\\
\vdots&&&&&\vdots&\\
\bfzr_{t\times i_2}  &\bfzr_{t\times t} &\cdots & \bfI_{t\times t} &\multicolumn{3}{c}{{b_{z-1}}\bfI_{t\times
		t}}\\
\bfzr_{(t-1)\times i_2}&\bfzr_{(t-1)\times t}&\cdots&\bfzr_{(t-1)\times t}&\multicolumn{3}{c}{\bfC(c_1, c_2)_{(t-1)\times t}}
\end{array}
\right]
\label{eq:typeII_D}
\end{align}
	\end{table*}

	Note that if $z_1=0$, $\bfC=[\bfzr_{i_1\times(k-t-i)}~b_1\bfI_{i_1\times i_1}~\bfzr_{i_1\times (i_2+1)}]$ and if $z_1=z-2$, 
	$$
	\mathbf D =
	\left[
	\begin{array}{cccc}
	\bfzr_{1\times i_2} & \bfzr_{1\times i_1} & b_{z-1} & \bfzr_{1\times i_2}\\
     \bfI_{i_2\times i_2} &\multicolumn{2}{c}{\bfzr_{i_2\times (i_1+1)}} & {b_{z-1}}\bfI_{i_2\times i_2}\\
	\bfzr_{(t-1)\times i_2}&\multicolumn{3}{c}{\bfC(c_1, c_2)_{(t-1)\times t}}
	\end{array}
	\right].
	$$
To verify $\bfG_{\calS_1 \setminus i}$ has full rank, we just need to check $\bfD$ has full rank (as $\mathbf A$ is full rank). Checking that $\bfD$ has full rank can be further simplified as follows. As $b_{z_1 + 1} \neq 0$, we can move the corresponding column that has $b_{z_1 + 1}$ as its first entry so that it is the first column of $\bfD$.

Following this, consider $\bfC(c_1, c_2)_{(t-1)\times t}\setminus \bfc_{i_1}$ which is obtained by deleting the $i_1$-th column of $\bfC(c_1, c_2)_{(t-1)\times t}$. 
	\begin{align*}
		\bfC(c_1, c_2)_{(t-1)\times t}\setminus \bfc_{i_1} = \begin{bmatrix}
			\bfD_1 & \bfD_2\\
			\bfD_3 & \bfD_4
		\end{bmatrix},
	\end{align*}
	where $\bfD_1$ is a $i_1\times i_1$ matrix as follows
	\begin{align*}
		\bfD_1=\begin{bmatrix}
			c_1 & c_2 & 0 & 0 &\cdots&0 & 0\\
			0 & c_1 & c_2 & 0 &\cdots&0 & 0\\
			\vdots&&&&&&\vdots\\
			0 & 0 & 0 & \cdots & 0 &c_1 & c_2\\
			0 & 0 & 0 & \cdots & 0 &0 & c_1
		\end{bmatrix},
	\end{align*}
	$\bfD_4$ is a $i_2\times i_2$ matrix as follows
	\begin{align*}
		\bfD_4=\begin{bmatrix}
			c_2 & 0 & 0 & \cdots & 0 &0 & 0\\
			c_1 & c_2 & 0 & \cdots &0&0 & 0\\
			\vdots&&&&&&\vdots\\
			0 & 0 & \cdots & 0 & c_1 &c_2 & 0\\
			0 & 0 & \cdots & 0 & 0 &c_1 & c_2
		\end{bmatrix},
	\end{align*}
	and $\bfD_2$ and $\bfD_3$ are $i_1\times i_2$ and $i_2\times i_1$ all zero matrices respectively. Then $\det(\bfC(c_1, c_2)_{(t-1)\times t}\setminus \bfc_{i_1})=c_1^{i_1}c_2^{i_2}$ and $\det(\bfG_{\calS_1\setminus i}) = \pm b_{z_1+1}c_1^{i_1}c_2^{i_2} \neq 0$. 
	
	Case 2: By deleting any of last $t$ columns in $\bfG_{\calS_1}$, say $i$-th column of $\bfG_{\calS_1}$, where $(z-1)t\le i<zt$, the block form of resultant matrix $\bfG_{\calS_1 \setminus i}$ can be expressed as follows.
	\begin{align*}
		\bfG_{\calS_1 \setminus i}=
		\left[
		\begin{array}{c|c}
			\mathbf A&\mathbf C\\  \hline
			\mathbf B&\mathbf D
		\end{array}
		\right],
	\end{align*}
	where $\bfA = \bfI_{(z-1)t\times (z-1)t}$, $\mathbf B=\bfzr_{(t-1)\times (z-1)t}$, $\bfC$ is obtained by deleting the $(i-(z-1)t)$-th column of matrix $[{b_1}\bfI_{t\times t}~{b_2}\bfI_{t\times t}~\cdots~{b_{z-1}}\bfI_{t\times t}]^T$ and $\bfD$ is $\bfC(c_1, c_2)_{(t-1)\times t}\setminus \bfc_{i-(z-1)t}$. Since $\det(\bfD)=c_1^{i-(z-1)t}c_2^{zt-i-1}\neq 0$,  $\det(\bfG_{\calS_1 \setminus i})= \pm c_1^{i-(z-1)t}c_2^{zt-i-1}\neq 0$ and therefore $\bfG_{\calS_1 \setminus i}$ has full rank. 
	
	\item Type III
	when $a>1$, $\bfG_{\calS_a}$ has the form in eq. (\ref{eq:claim_6_type_3}). 
	\begin{table*}
	\begin{align}
		\bfG_{\calS_a}=
		\begin{bmatrix}
		\underbrace{
		\begin{array}{ccc}
		\bfI_{t\times t}& \cdots& \bfzr_{t\times t}\\
		\vdots&&\vdots\\
		\bfzr_{t\times t}& \cdots& \bfI_{t\times t}\\
		\bfzr_{t\times t}& \cdots& \bfzr_{t\times t}\\
		\bfzr_{t\times t}& \cdots& \bfzr_{t\times t}\\
		\vdots&&\vdots\\
		\bfzr_{t\times t}& \cdots& \bfzr_{t\times t}\\
		\bfzr_{(t-1)\times t}& \cdots& \bfzr_{(t-1)\times t}	
		\end{array}
	    }_{\text{Case 1}}
	    &
	    \underbrace{
	    	\begin{array}{ccc}
	    	\bfzr_{t\times t} &\cdots&\bfzr_{t\times t}\\
	    	\vdots&&\vdots\\
	    	\bfzr_{t\times t} &\cdots&\bfzr_{t\times t}\\
	    	\bfzr_{t\times t} &\cdots&\bfzr_{t\times t}\\
	    	\bfI_{t\times t} &\cdots&\bfzr_{t\times t}\\
	    	\vdots&&\vdots\\
	    	\bfzr_{t\times t}& \cdots& \bfI_{t\times t}\\
	    	\bfzr_{(t-1)\times t} &\cdots&\bfzr_{(t-1)\times (t-1)}	
	    	\end{array}
	    }_{\text{Case 2}}
	    &
	    \underbrace{
	    	\begin{array}{c}
	    	\bfzr_{t\times (t-1)}\\
	    	\vdots\\
	    	\bfzr_{t\times (t-1)}\\
	    	\bfzr_{t\times (t-1)}\\
	    	\bfzr_{t\times (t-1)}\\
	    	\vdots\\
	    	\bfzr_{t\times (t-1)}\\
	    	\bfI_{(t-1)\times (t-1)}
	    	\end{array}
	    	}_{\text{Case 3}}
	    &
	    \underbrace{
	    	\begin{array}{c}
	    	\bfoe_{t}\\
	    	\vdots\\
	    	\bfoe_{t}\\
	    	\bfoe_{t}\\
	    	\bfoe_{t}\\
	    	\vdots\\
	    	\bfoe_{t}\\
	    	\bfoe_{t-1}
	    	\end{array}
	    }_{\text{Case 4}}
	    &
	    	\underbrace{
	    		\begin{array}{c}
	    		{b_1}\bfI_{t\times t}\\
	    		\vdots\\
	    		{b_{z-a}}\bfI_{t\times t}\\
	    		{b_{z-a+1}}\bfI_{t\times t}\\
	    		{b_{z-a+2}}\bfI_{t\times t}\\
	    		\vdots\\
	    		{b_{z-1}}\bfI_{t\times t}\\
	    		\bfC(c_1, c_2)_{(t-1)\times t}
	    		\end{array}
	    	}_{\text{Case 5}}
		\end{bmatrix}
    \label{eq:claim_6_type_3}
	\end{align}
	\end{table*}
As before we perform a case analysis. Each of the cases is specified by the corresponding underbrace in eq. (\ref{eq:claim_6_type_3}).

    Case 1:	By deleting the $i$-th column of $\bfG_{\calS_a}$, where $z_1t\le i< (z_1+1)t$, $z_1\le z-a-1$, $i_1=i-z_1t$, and $i_2=(z_1+1)t-i-1$, the block form of the resultant matrix $\bfG_{\calS_a \setminus i}$ can be expressed as follows,
	
	\begin{align*}
		\bfG_{\calS_a \setminus i}=
		\left[
		\begin{array}{c|c}
			\mathbf A&\mathbf C\\  \hline
			\mathbf B&\mathbf D
		\end{array}
		\right],
	\end{align*}
	where $\bfA= \bfI_{i\times i}, \mathbf B= \bfzr_{(k-i)\times i}
		$, $\mathbf C$ and $\mathbf D$ has the form in eq. (\ref{eq:typeIII_C}) and eq. (\ref{eq:typeIII_D}), respectively.
		\begin{table*}[t]
		\begin{align}
	\mathbf C &=
\left[
\begin{array}{ccccccc}
\bfzr_{t\times i_2} &\bfzr_{t\times t} &\cdots & \bfzr_{t\times (t-1)} & \bfoe_t &\multicolumn{2}{c}{{b_1}\bfI_{t\times t}}\\
\vdots&&&&&\multicolumn{2}{c}{\vdots}\\
\bfzr_{t\times i_2} &\bfzr_{t\times t}&\cdots & \bfzr_{t\times (t-1)} & \bfoe_t &\multicolumn{2}{c}{{b_{z_1}}\bfI_{t\times t}}\\
\bfzr_{i_1\times i_2}&\bfzr_{i_1\times t} &\cdots & \bfzr_{i_1\times (t-1)} & \bfoe_{i_1} &{b_{z_1+1}}\bfI_{i_1\times i_1}&\bfzr_{i_1\times(i_2+1)}
\label{eq:typeIII_C}
\end{array}
\right]
		\end{align}
		\end{table*}
			\begin{table*}[t]
	\begin{align}
	\bfD&=	\left[
\begin{array}{cccccccccc}
\bfzr_{1\times i_2} &\bfzr_{1\times t} &\cdots&\cdots &\cdots  &\bfzr_{1\times (t-1)} & 1 &\bfzr_{1\times i_1}  & b_{z_1+1} & \bfzr_{1\times i_2} \\
\bfI_{i_2\times i_2} &\bfzr_{i_2\times t} &\cdots&\cdots&\cdots  &\bfzr_{i_2\times (t-1)} &\bfoe_{i_2} &\multicolumn{2}{c}{\bfzr_{i_2\times (i_1+1)}} & {b_{z_1+1}}\bfI_{i_2\times i_2}\\
\bfzr_{t\times i_2} &\bfI_{t\times t} &\cdots&\cdots&\cdots  &\bfzr_{t\times (t-1)} &\bfoe_{t} &\multicolumn{3}{c}{{b_{z_1+2}}\bfI_{t\times t}}\\
\vdots &&&&&\vdots\\
\bfzr_{t\times i_2}& \bfzr_{t\times t}& \cdots &\bfzr_{t\times t} &\cdots&\bfzr_{t\times (t-1)}&\bfoe_t&\multicolumn{3}{c}{{b_{z-a+1}}\bfI_{t\times t}}\\
\bfzr_{t\times i_2}& \bfzr_{t\times t}& \cdots &\bfI_{t\times t} &\cdots&\bfzr_{t\times (t-1)}&\bfoe_t&\multicolumn{3}{c}{{b_{z-a+2}}\bfI_{t\times t}}\\
\vdots&&&&&&&&&\vdots\\
\bfzr_{(t-1)\times i_2}& \bfzr_{(t-1)\times t}& \cdots &\bfzr_{(t-1)\times t} &\cdots&\bfI_{(t-1)\times (t-1)}&\bfoe_{t-1}&\multicolumn{3}{c}{\bfC(c_1, c_2)_{(t-1)\times t}}
\label{eq:typeIII_D}
\end{array}
\right]
	\end{align}
\end{table*}
	Note that if $z_1=0$, $\bfC=[\bfzr_{i_1\times i_2}~\bfzr_{i_1\times t}~\cdots~ \bfzr_{i_1\times (t-1)}~\bfoe_{i_1}~{b_1}\bfI_{i_1\times i_1}~\bfzr_{i_1\times(i_2+1)}]$ and if $z_1=z-a-1$, $\mathbf D$ has the form in (\ref{eq:typeIII_D_2}).
	\begin{table*}[h]
	\begin{align}
	\mathbf D =
	\left[
	\begin{array}{cccccccccc}
	\bfzr_{1\times i_2} &\bfzr_{1\times t} &\cdots&\cdots &\cdots  &\bfzr_{1\times (t-1)} & 1 &\bfzr_{1\times i_1}  & b_{z-a} & \bfzr_{1\times i_2} \\
	\bfI_{i_2\times i_2} &\bfzr_{i_2\times t} &\cdots&\cdots&\cdots  &\bfzr_{i_2\times (t-1)} &\bfoe_{i_2} &\multicolumn{2}{c}{\bfzr_{i_2\times (i_1+1)}} & {b_{z-a}}\bfI_{i_2\times i_2}\\
	\bfzr_{t\times i_2}& \bfzr_{t\times t}& \cdots &\bfzr_{t\times t} &\cdots&\bfzr_{t\times (t-1)}&\bfoe_t&\multicolumn{3}{c}{{b_{z-a+1}}\bfI_{t\times t}}\\
	\bfzr_{t\times i_2}& \bfzr_{t\times t}& \cdots &\bfI_{t\times t} &\cdots&\bfzr_{t\times (t-1)}&\bfoe_t&\multicolumn{3}{c}{{b_{z-a+2}}\bfI_{t\times t}}\\
	\vdots&&&&&&&&&\vdots\\
	\bfzr_{(t-1)\times i_2}& \bfzr_{(t-1)\times t}& \cdots &\bfzr_{(t-1)\times t} &\cdots&\bfI_{(t-1)\times (t-1)}&\bfoe_{t-1}&\multicolumn{3}{c}{\bfC(c_1, c_2)_{(t-1)\times t}}
	\end{array}
	\right].
	\label{eq:typeIII_D_2}
	\end{align}
	\end{table*}

	To verify that $\bfG_{\calS_a \setminus i}$ has full rank, we just need to check $\bfD$ has full rank. Owing to the construction of $\bfD$, we have to check the determinant of the following $(t+1) \times (t+1)$ matrix.
	\begin{align*}
	\bfF=
		\begin{bmatrix}
			1&0 & 0 & \cdots & b_{z_1+1}& \cdots &0\\
			1&b_{z-a+1} &0 &\cdots & 0 & \cdots &0\\
			1&0 & b_{z-a+1}&\cdots & 0 & \cdots &0\\
			\vdots&&&&\vdots\\
			1&0&0&\cdots & 0 & \cdots&b_{z-a+1}
		\end{bmatrix};
	\end{align*}
	$\det(\bfF)=(b_{z-a+1}-b_{z_1+1})b_{z-a+1}^{t-1}$. Since $z_1\neq z-a$ and then $b_{z_1+1}\neq b_{z-a+1}$, the above matrix has full rank and $\det(\bfG_{\calS_a \setminus i})= \pm (b_{z-a+1}-b_{z_1+1})b_{z-a+1}^{t-1}\neq 0$, so that $\bfG_{\calS_a \setminus i}$ has full rank.
	
	Case 2: By deleting $i$-th column of $\bfG_{\calS_a}$, where $z_1t\le i< (z_1+1)t$, $z-a\le z_1\le z-3$, the proof that the resultant matrix has full rank is similar to the case that $z_1\le z-a-1$ and we omit it here. 
	
	Case 3: By deleting $i$-th column of $\bfG_{\calS_a}$, where $(z-2)t\le i\le (z-1)t-2$, $i_1 = i-(z-2)t$ and $i_2=(z-1)t-2-i$, the resultant matrix is as follows,
	\begin{align*}
		\bfG_{\calS_a \setminus i}=
		\left[
		\begin{array}{c|c}
			\mathbf A&\mathbf C\\  \hline
			\mathbf B&\mathbf D
		\end{array}
		\right],
	\end{align*}
	where \begin{align*}
		\bfA&= \bfI_{(z-a)t\times (z-a)t}\\
		\mathbf B &= \bfzr_{(k-(z-a)t)\times (z-a)t}\\
		\mathbf C &= \begin{bmatrix}
			\bfzr_{t\times t}&\cdots&\bfzr_{t\times (t-2)}&\bfoe_t&{b_1}\bfI_{t\times t}\\
			\vdots&&&&\vdots\\
			\bfzr_{t\times t}&\cdots&\bfzr_{t\times (t-2)}&\bfoe_t&{b_{z-a}}\bfI_{t\times t}
		\end{bmatrix}\\	 	  	
		\bfD&=\begin{scriptsize}
	\left[
		\begin{array}{cccccc}
			\bfzr_{t\times t} & \cdots &\multicolumn{2}{c}{\bfzr_{t\times (t-2)}} & \bfoe_t & {b_{z-a+1}}\bfI_{t\times t}\\
			\bfI_{t\times t} & \cdots & \multicolumn{2}{c}{\bfzr_{t\times (t-2)}} & \bfoe_t & {b_{z-a+2}}\bfI_{t\times t}\\
			\vdots&&&&\vdots\\
			\multirow{3}*{$\bfzr_{(t-1)\times t}$}&\multirow{3}*{$\cdots$}&\bfI_{i_1\times i_1}&\bfzr_{i_1\times i_2}&\multirow{3}*{$\bfoe_{t-1}$}&\multirow{3}*{$\bfC(c_1, c_2)_{(t-1)\times t}$}\\
			&&\bfzr_{1\times i_1}&\bfzr_{1\times i_2}&&\\
			&&\bfzr_{i_2\times i_1}&\bfI_{i_2\times i_2}&&
		\end{array}
		\right]
				\end{scriptsize}
	\end{align*}
	To verify $\bfG_{\calS_a\setminus i}$ has full rank, we need to check the determinant of $\bfD$. Owing to the construction of $\bfD$, the following matrix is required to be full rank,
	\begin{align*}
		\bfD'=&\begin{bmatrix}
			\bfoe_t&{b_{z-a+1}}\bfI_{t\times t}\\
			1&\bfC(c_1, c_2)_{(t-1)\times t}(i_1)
		\end{bmatrix}\\=&
		\begin{bmatrix}
		\begin{smallmatrix}
			1&b_{z-a+1}&0&\cdots&0&0&0&\cdots&0\\
			1&0&b_{z-a+1}&\cdots&0&0&0&\cdots&0\\
			\vdots&&&&&&&&\vdots\\
			1&0&0&\cdots&0&0&0&\cdots&b_{z-a+1}\\
			1&0&\cdots&0&c_1&c_2&0&\cdots&0
			\end{smallmatrix}
		\end{bmatrix},
	\end{align*}
	where $\bfC(c_1, c_2)_{(t-1)\times t}(i_1)$ denotes the $i_1$-th row of $\bfC(c_1, c_2)_{(t-1)\times t}$, $0\le i\le t-2$.
	\begin{align*}
	\det \bfD'=&\det (b_{z-a+1}\bfI_{t\times t})\cdot \det(1-\\
	&\bfC(c_1,c_2)_{(t-1)\times t}(i_1)\cdot(b^{-1}_{z-a+1}\bfI_{t\times t})\cdot\bfoe_t)\\
	=&b_{z-a+1}^t(1-b_{z-a+1}^{-1}(c_1+c_2))
	\end{align*}
	
	Since $b_{z-a+1}\neq 0$ and $c_1+c_2=0$, $\det \bfD'\neq 0$ and $\bfD'$ has full rank. Then $\det(\bfD)=b_{z-a+1}^t(1-b_{z-a+1}^{-1}(c_1+c_2))\neq 0$ and  thus $\bfG_{\calS_a\setminus i}$ is full rank.

	Case 4: By deleting $i$-th column of $\bfG_{\calS_a}$, where $i=(z-1)t-1$, the block form of the resultant matrix $\bfG_{\calS_a \setminus i}$ can be expressed as eq. (\ref{eq:case4_G1}).
	\begin{table*}[t]\begin{align}
		\bfG_{\calS_a\setminus i}&=
		\begin{bmatrix}
			\bfI_{t\times t}& \cdots& \bfzr_{t\times t} &\bfzr_{t\times t} &\cdots&\bfzr_{t\times (t-1)}&{b_1}\bfI_{t\times t}\\
			\vdots&&&&&&\vdots\\
			\bfzr_{t\times t}& \cdots& \bfI_{t\times t} &\bfzr_{t\times t} &\cdots&\bfzr_{t\times (t-1)}&{b_{z-a}}\bfI_{t\times t}\\
			\bfzr_{t\times t}& \cdots& \bfzr_{t\times t} &\bfzr_{t\times t} &\cdots&\bfzr_{t\times (t-1)}&{b_{z-a+1}}\bfI_{t\times t}\\
			\bfzr_{t\times t}& \cdots& \bfzr_{t\times t} &\bfI_{t\times t} &\cdots&\bfzr_{t\times (t-1)}&{b_{z-a+2}}\bfI_{t\times t}\\
			\vdots&&&&&&\vdots\\
			\bfzr_{(t-1)\times t}& \cdots& \bfzr_{(t-1)\times t} &\bfzr_{(t-1)\times t} &\cdots&\bfI_{(t-1)\times (t-1)}&\bfC(c_1, c_2)_{(t-1)\times t}
		\end{bmatrix}
			\label{eq:case4_G1}
	\end{align}
	\end{table*}
   Evidently, $\det(\bfG_{\calS_a\setminus i})= \pm b_{z-a+1}^t$, so that $\bfG_{\calS_a\setminus i}$  has full rank.
	
	Case 5: By deleting $i$-th column of $\bfG_{\calS_a}$, where $(z-1)t\le i<zt$ and $i_1=i-(z-1)t$, the block form of the resultant matrix $\bfG_{\calS_a \setminus i}$ can be expressed as eq. (\ref{eq:case4_G2}).
	
	\begin{table*}[t]
	\begin{align}
		\bfG_{\calS_a\setminus i}&=
		\begin{bmatrix}
			\bfI_{t\times t}& \cdots& \bfzr_{t\times t} &\bfzr_{t\times t} &\cdots&\bfzr_{t\times (t-1)}&\bfoe_t&{b_1}\bfI_{t\times t}\setminus \bfc_{i_1}\\
			\vdots&&&&&&&\vdots\\
			\bfzr_{t\times t}& \cdots& \bfI_{t\times t} &\bfzr_{t\times t} &\cdots&\bfzr_{t\times (t-1)}&\bfoe_t&{b_{z-a}}\bfI_{t\times t}\setminus \bfc_{i_1}\\
			\bfzr_{t\times t}& \cdots& \bfzr_{t\times t} &\bfzr_{t\times t} &\cdots&\bfzr_{t\times (t-1)}&\bfoe_t&{b_{z-a+1}}\bfI_{t\times t}\setminus \bfc_{i_1}\\
			\bfzr_{t\times t}& \cdots& \bfzr_{t\times t} &\bfI_{t\times t} &\cdots&\bfzr_{t\times (t-1)}&\bfoe_t&{b_{z-a+2}}\bfI_{t\times t}\setminus \bfc_{i_1}\\
			\vdots&&&&&&&\vdots\\
			\bfzr_{(t-1)\times t}& \cdots& \bfzr_{(t-1)\times t} &\bfzr_{(t-1)\times t} &\cdots&\bfI_{(t-1)\times (t-1)}&\bfoe_{t-1}&\bfC(c_1, c_2)_{(t-1)\times t}\setminus \bfc_{i_1}
		\end{bmatrix}
				\label{eq:case4_G2}
	\end{align}
\end{table*}
	where ${b_s}\bfI_{t\times t}\setminus \bfc_{i_1}$ denotes the submatrix obtained by deleting $i_1$-th column of ${b_s}\bfI_{t\times t}$ and $\bfC(c_1, c_2)_{(t-1)\times t}\setminus \bfc_{i_1}$ denotes the submatrix obtained by deleting $i_1$-th column of $\bfC(c_1, c_2)_{(t-1)\times t}\setminus \bfc_{i_1}$. To verify $\bfG_{\calS_a\setminus i}$ has full rank, we just need to check $[\bfoe_t|{b_{z-a+1}}\bfI_{t\times t}\setminus \bfc_{i_1}]$ has full rank. Since $[\bfoe_{t}|{b_{z-a+1}}\bfI_{t\times t}]$ has the following form,
	\begin{align*}
		\begin{bmatrix}
			1&b_{z-a+1} &0 &\cdots &0\\
			1&0 & b_{z-a+1}&\cdots &0\\
			\vdots&&&&\vdots\\
			1&0&0&\cdots&b_{z-a+1}
		\end{bmatrix},
	\end{align*}
	by deleting any column of above matrix, it is obvious that $\det([\bfoe_t|{b_{z-a+1}}\bfI_{t\times t}\setminus \bfc_{i_1}])=\pm b_{z-a+1}^{t-1}$ and $\det(\bfG_{\calS_a\setminus i})\neq 0$.
	
\end{itemize}

\subsection*{Proof of Claim \ref{claim:matrixextend}}
Let $z$ be the least integer such that $k+1~|~nz$. First, we argue that $z$ is the least integer such that $k+1~|~(n+s(k+1))z$. Assume that this is not true, then there exists $z'<z$ such that $k+1~|~(n+s(k+1))z'$. As $n\ge k+1$ and $k+1~|~s(k+1)z'$ this implies that $k+1~|~nz'$ which is a contradiction.

Next we argue that $\bfG'$ satisfies the CCP, i.e., all $k \times k$ submatrices of each $\bfG'_{\calS'_a}$, where $\calT'_a=\{a(k+1),\cdots, a(k+1)+k\}$ and $\calS'_a=\{(t)_{n+s(k+1)}|t\in \calT'_a\}$ and $0\le a\le \frac{nz}{k+1}+sz$, are full rank. Let $n'= n+s(k+1)$. We argue it in three cases.

\begin{itemize}
\item {\it Case 1. The first column of $\bfG'_{\calS'_a}$ lies in the first $s(k+1)$ columns of $\bfG'$.}\\
Suppose $ln'\le a(k+1)< ln'+s(k+1)$ where $0\le l<  z-1$. By the construction of $\bfG'$, $\bfG'_{\calS_a}=\bfD$. Since $\bfD = \bfG_{\calS_0}$, all $k\times k$ submatrices of $\bfD$ have full rank and so does $\bfG'_{\calS_a}$.

\item {\it Case 2. The first and last column of $\bfG'_{\calS'_a}$ lie in the last $n$ columns of $\bfG'$.}\\
Suppose  $ln'+s(k+1)\le a(k+1)$ and $ a(k+1)+k< (l+1)n'$ where $0\le l< z-1$. As $n'>s(k+1)$, $a(k+1)-(l+1)s(k+1)>0$ and $k+1|a(k+1)-(l+1)s(k+1)$. Let $a'=a-(l+1)s$, then $\bfG'_{\calS_a}=\bfG_{\calS_{a'}}$ and hence  all $k\times k$ submatrices of $\bfG'_{\calS_a}$ have full rank.

\item {\it Case 3. The first column of $\bfG'_{\calS'_a}$ lies in the last $n$ columns of $\bfG'$ but the last column lies in the first $(k+1)$ columns of $\bfG'$.}\\
Suppose  $ln'+s(k+1)\le a(k+1)$ and $ a(k+1)+k> (l+1)n'$ where $0\le l<z-2$. Again, we can get $k+1|a(k+1)-(l+1)s(k+1)$ and let $a'=a-(l+1)s$. Let $\calS'^1=\{(a(k+1))_{n'},\cdots, (ln'+n'-1)_{n'}\}$ and $\calS^1=\{(a'(k+1))_n,\cdots,(ln+n-1)_n\}$. As $(ln'+n'-1)-a(k+1)=(ln+n-1)-a'(k+1)$, $\bfG'_{\calS'^1}=\bfG_{\calS^1}$. Let $\calS'^2=\{(ln'+n')_{n'},\cdots,(a(k+1)+k)_{n'}\}$ and $\calS^2=\{(ln+n)_n,\cdots, (a'(k+1)+k)_n\}$. By the construction of $\bfG'$, $\bfG_{\calS^2}=\bfG'_{\calS'^2}$. Then $\bfG'_{\calS_a}=[\bfG'_{\calS'^1} \bfG'_{\calS'^2}]=[\bfG_{\calS^1} \bfG_{\calS^2}]=\bfG_{\calS_{a'}}$ and hence all $k\times k$ submatrices of $\bfG'_{\calS_a}$ have full rank.
\end{itemize}

\subsection*{Proof of Claim \ref{claim:compare_MN}}

\begin{itemize}
	\item $\frac{M}{N}=\frac{1}{q}$.
We have
\begin{align*}
\frac{1}{K} \log_2 \frac{F_s^{MN}}{F_s^*} &= \frac{1}{K} \log_2 \binom{K}{K/q} - \frac{1}{K} \log_2 z - \frac{k}{K} \log_2 q.
\end{align*}
Using the fact that $z \leq k+1$ and taking limits as $n \rightarrow \infty$, we get that
\begin{align*}
\lim_{n \rightarrow \infty} \frac{1}{K} \log_2 \frac{F_s^{MN}}{F_s^*} = H_2\bigg{(}\frac{1}{q}\bigg{)} - \frac{\eta}{q} \log_2 q.
\end{align*}

\item $\frac{M}{N}=1-\frac{k+1}{nq}$. We have 
\begin{align*}
\frac{1}{K} \log_2 \frac{F_s^{MN}}{F_s^*} =& \frac{1}{K}\log_2 \binom{K}{k+1}-\frac{k+1}{K}\log_2 q\\
&-\frac{1}{K}\log_2\frac{zn}{k+1}.
    \end{align*}
 Using the fact that $z \leq k+1$ and taking limits as $n \rightarrow \infty$, we get that
   	\begin{align*}
  \lim_{n \rightarrow \infty}  	\frac{1}{K} \log_2 \frac{F_s^{MN}}{F_s^*}&=H_2\bigg{(}\frac{\eta}{q}\bigg{)} - \frac{\eta}{q} \log_2 q.
   	\end{align*}
\end{itemize}

\subsection{Discussion on coded caching systems constructed by generator matrices satisfying the $(k,\alpha)$-CCP where $\alpha \leq k$}
\label{sec:kalpha_ccp_matrices}
Consider the $(k,\alpha)$-CCP ({\it cf.} Definition \ref{def:kalphacc}) where $\alpha\le k$. Let $z$ be the least integer such that $\alpha~|~nz$, and let $\calT^\alpha_a=\{a\alpha,\cdots,a\alpha+\alpha-1)\}$ and $\calS^\alpha_a=\{(t)_n~|~t\in \calT^\alpha_a\}$. Let $\bfG_{\calS^\alpha_a}=[\bfg_{i_0},\cdots,\bfg_{i_{\alpha-1}}]$ be the submatrix of $\bfG$ specified by the columns in $\calS^\alpha_a$, i.e,  $\bfg_{i_j} \in \bfG_{\calS^\alpha_a}$ if $i_j\in \calS^\alpha_a$.
We demonstrate that the resolvable design generated from a linear block code that satisfies the $(k,\alpha)$-CCP can also be used in a coded caching scheme. First, we construct a $(X,\calA)$ resolvable design as described in Section \ref{sec:design}.A., which can be partitioned into $n$ parallel classes $\calP_i=\{B_{i,j}: 0\le j<q\}$, $0\le i<n$. By the constructed resolvable design, we partition each subfile $W_n$ into $q^kz$ subfiles $W_n=\{W_{n,t}^s~|~0\le t< q^k, 0\le s< z\}$ and operate the placement scheme in Algorithm \ref{Alg:Placement}. In the delivery phase, for each recovery set, several equations are generated, each of which benefit $\alpha$ users simultaneously. Furthermore, the equations generated by all the recovery sets can recover all the missing subfiles. In this section, we only show that for the recovery set $\calP_{\calS^\alpha_a}$, it is possible to generate equations which benefit $\alpha$ users and allow the recovery of all of missing subfiles with given superscript. The subsequent discussion exactly mirrors the discussion in the $(k,k+1)$-CCP case and is skipped.

Towards this end, we first show that picking $\alpha$ users from $\alpha$ distinct parallel classes can always form $q^{k-\alpha+1}-q^{k-\alpha}$ signals. More specifically, consider blocks $B_{i_1, l_{i_1}}, \dots, B_{i_{\alpha}, l_{i_{\alpha}}}$ (where $l_{i_j} \in \{0, \dots, q-1\}$) that are picked from $\alpha$ distinct parallel classes of $\calP_{\calS^\alpha_{a}}$. Then, $|\cap_{j=1}^{\alpha-1} B_{i_j, l_{i_j}}| = q^{k-\alpha+1}$ and $|\cap_{j=1}^\alpha B_{i_j, l_{i_j}}| = q^{k-\alpha}$.

\begin{claim}
	\label{claim:MDSintersectionKAlpha}
	Consider the resolvable design $(X, \calA)$ constructed by a $(n,k)$ linear block code that satisfies the $(k,\alpha)$ CCP. Let $\calP_{\calS^\alpha_a}=\{\calP_i~|~i\in \calS^\alpha_a\}$ for $0\le a< \frac{zn}{\alpha}$, i.e., it is the set of parallel classes corresponding to $\calS^\alpha_a$. We emphasize that $|\calP_{\calS^\alpha_a}| = \alpha\le k$. Consider blocks $B_{i_1, l_{i_1}}, \dots, B_{i_{\alpha'}, l_{i_{\alpha'}}}$ (where $l_{i_j} \in \{0, \dots, q-1\}$) that are picked from any $\alpha'$ distinct parallel classes of $\calP_{\calS^\alpha_a}$ where $\alpha'\le \alpha$. Then, $|\cap_{j=1}^{\alpha'} B_{i_j, l_{i_j}}| = q^{k-\alpha'}$.
\end{claim}
	The above argument implies that any $\alpha-1$ blocks from any $\alpha-1$ distinct parallel classes of $\calP_{\calS^\alpha_a}$ have $q^{k-\alpha+1}$ points in common and any $\alpha$ blocks $B_{i_1,l_{i_1}}$, $B_{i_\alpha,l_{i_\alpha}}$ from any $\alpha$ distinct parallel classes of $\calP_{\calS^{\alpha}_a}$ have $q^{k-\alpha}$ points in common. These blocks (or users) can participate in $q^{k-\alpha+1}-q^{k-\alpha}$ equations, each of which benefits $\alpha$ users. In particular, each user will recover a missing subfile indexed by an element belonging to the intersection of the other $\alpha-1$ blocks in each equation. A very similar argument to Lemma \ref{lemma:delivery} can be made to justify enough equations can be found that allow all users to recover all missing subfiles. 
\begin{proof}
	Recall that by the construction in Section III.A,  block $B_{i,l} \in \calP_i$ is specified as follows,
	$$
	B_{i,l} = \{j : \bfT_{i,j} = l\}.
	$$
	Let $\bfG=[g_{ab}]$, for $0\le a<k$, $0\le b<n$.
	
	Now consider $B_{i_1, l_{i_1}}, \dots, B_{i_{\alpha'}, l_{i_{\alpha'}}}$ (where $i_j \in \calS^\alpha_a, l_{i_j} \in \{0, \dots, q-1\}$) that are picked from $\alpha'$ distinct parallel classes of  $\calP_{\calS^\alpha_a}$. W.l.o.g. we assume that $i_1 < i_2 < \dots < i_{\alpha'}$. Let $\calI =  \{i_1, \dots, i_{\alpha'}\}$ and $\bfT_{\calI}$ denote the submatrix of $\bfT$ obtained by retaining the rows in $\calI$. We will show that the vector $[l_{i_1}~l_{i_2}~\dots~l_{i_{\alpha'}}]^T$ is a column in $\bfT_{\calI}$ and appears $q^{k-\alpha'}$ times in it.
	
	We note here that by the $(k,\alpha)$-CCP, the vectors $\bfg_{i_1}, \bfg_{i_2}, \ldots, \bfg_{i_\alpha}$ are linearly independent and thus the subset of these vectors, $\bfg_{i_1},\cdots,\bfg_{i_{\alpha'}}$ are linearly independent. W. l. o. g., we assume that the top $\alpha' \times \alpha'$ submatrix of the matrix $[\bfg_{i_1}~\bfg_{i_2}~ \dots ~\bfg_{i_{\alpha'}}]$ is full-rank.  Next, consider the system of equations in variables $\bfu_0, \dots, \bfu_{\alpha'-1}$.
	\begin{align*}
	\sum_{b=0}^{\alpha'-1}\bfu_{b}g_{bi_1} &= l_{i_1}-\sum_{b=\alpha'}^{k-1}\bfu_{b}g_{bi_1},\\
	\sum_{b=0}^{\alpha'-1}\bfu_{b}g_{bi_2} &= l_{i_2}-\sum_{b=\alpha'}^{k-1}\bfu_{b}g_{bi_2},\\
	\mathrel{\makebox[\widthof{=}]{\vdots}}\\
	\sum_{b=0}^{\alpha'-1}\bfu_{b}g_{bi_{\alpha'}} &= l_{i_{\alpha'}}-\sum_{b=\alpha'}^{k-1}\bfu_{b}g_{bi_{\alpha'}}.
	\end{align*}
By the assumed condition, it is evident that this system of $\alpha'$ equations in $\alpha'$ variables has a unique solution for a given vector $\bfv=[\bfu_{\alpha'},\cdots, \bfu_{k-1}]$ over $GF(q)$. Since there are $q^{k-\alpha'}$ possible $\bfv$ vectors, the result follows.
	\end{proof}
		
As in the case of the $(k,k+1)$-CCP, we form a recovery set bipartite graph with parallel classes and recovery sets as the disjoint vertex subsets, and the edges incident on each parallel class are labeled arbitrarily from $0$ to $z-1$. For a parallel class $\calP\in \calP_{\calS^{\alpha}_a}$ we denote this label by label($\calP-\calP_{\calS^{\alpha}_a}$). For a given recovery set $\calP_{\calS^{\alpha}_a}$, the delivery phase proceeds by choosing blocks from $\alpha$ distinct parallel classes in $\calP_{\calS^{\alpha}_a}$ and it provides $q^{k-\alpha+1}-q^{k-\alpha}$ equations that benefit $\alpha$ users. Note that in the $(k,\alpha)$-CCP case, randomly picking $\alpha$ blocks from $\alpha$ parallel classes in $\calP_{\calS^{\alpha}_a}$ will always result in $q^{k-\alpha}$ intersections, which is different from $(k,k+1)$-CCP. It turns out that each equation allows a user in $\calP\in \calP_{\calS^{\alpha}_a}$ to recover a missing subfile with superscript label($\calP-\calP_{\calS_{a}^\alpha}$).

Let the demand of user $U_{B_{i,j}}$ for $i\in n-1, 0\le j\le q-1$ by $W_{\kappa_{i,j}}$. We formalize the argument in Algorithm \ref{Alg:SignalAlphaCCP} and prove that equations generated in each recovery set $\calP_{\calS^{\alpha}_a}$ can recover all missing subfile with superscript label$(\calP-\calP_{\calS^{\alpha}_a})$.

	\begin{algorithm}[htb]
		\SetNoFillComment
		\caption{Signal Generation Algorithm for $\calP_{\calS_a^\alpha}$}
		\label{Alg:SignalAlphaCCP}
		\SetKwInOut{Input}{Input}
		\SetKwInOut{Output}{Output}
		\Input{For $\calP \in \calP_{\calS_a^\alpha}$, $E(\calP) = \text{label}(\calP - \calP_{\calS_a^\alpha})$. Signal set $Sig=\emptyset$.}
		\While{any user $U_B\in \calP_j, j \in \calS^\alpha_a$ does not recover all its missing subfiles with superscript $E(\calP)$}
		{ Pick blocks $B_{j,l_j} \in \calP_j$ for all $j \in \calS_a^{\alpha}$ and $l_j \in \{0, \dots, q-1\}$\;
			\tcc{Pick blocks from distinct parallel classes in $\calP_{\calS_a^\alpha}$. The cardinality of their intersection is always $q^{k-\alpha}$}
			Find set $\hat{L}_{s} = \cap_{j \in \calS^\alpha_a \setminus \{s\}}B_{j,l_j}\setminus \cap_{j \in \calS^\alpha_a} B_{j,l_j}$ for $s\in \calS^\alpha_a$\;
			\tcc{Determine the missing subfile indices that the user from $\calP_s^\alpha$ will recover. Note that $|\hat{L}_{s}|=q^{k-\alpha+1}-q^{k-\alpha}$}
			Add signals $\oplus_{s \in S^\alpha_a} W^{E(\calP_s)}_{\kappa_{s,l_s},\hat{L}_s[t]}$, $0\le t< q^{k-\alpha+1}-q^{k-\alpha}$, to $Sig$\;
		\tcc{User $U_{B_{s,l_s}}$ demands file $W_{\kappa_{s,l_s}}$. This equation allows it to recover the corresponding missing subfile index $\hat{L}_s[t]$, which is the $t$-th element of $\hat{L}_s[t]$. The superscript is determined by the recovery set bipartite graph}
	}
		\Output{Signal set $Sig$.}
	\end{algorithm}
	
	For the sake of convenience we argue that user $U_{B_{\beta, l_\beta}}$ that demands $W_{\kappa_{\beta, l_\beta}}$ can recover all its missing subfiles with superscript $E(\calP_\beta)$. Note that $B_{\beta,l_{\beta}}=q^{k-1}$. Thus user $U_{B_{\beta,l_{\beta}}}$ needs to obtain $q^k-q^{k-1}$ missing subfiles with superscript $E(\calP_{\beta})$. The delivery phase scheme repeatedly picks $\alpha$ users from different parallel classes of $\calP_{\calS^\alpha_a}$. The equations in Algorithm \ref{Alg:SignalAlphaCCP} allow $U_{B_{\beta, l_\beta}}$ to recover all $W^{E(\calP_\beta)}_{\kappa_{\beta,l_\beta},\hat{L}_\beta[t]}$ where $\hat{L}_{\beta} = \cap_{j \in \calS^\alpha_a \setminus \{\beta\}}B_{j,l_j}\setminus \cap_{j \in \calS^\alpha_a} B_{j,l_j}$ and $t=1,\cdots, q^{k-\alpha+1}-q^{k-\alpha}$. This is because of Claim	\ref{claim:MDSintersectionKAlpha}.
	
	Next, we count the number of equations that  $U_{B_{\beta, l_\beta}}$  participates in. We can pick $\alpha-1$ users from $\alpha -1 $ parallel classes in $\calP_{\calS^{\alpha}_a}$. There are totally $q^{\alpha-1}$ ways to pick them, each of which generate $q^{k-\alpha+1}-q^{k-\alpha}$ equations. Thus there are a total of $q^k-q^{k-1}$ equations in which user  $U_{B_{\beta, l_\beta}}$ participates in.
	
	It remains to argue that each equation provides a distinct file part of user $U_{B_{\beta, l_\beta}}$. Towards this end, let $\{i_1,\cdots, i_{\alpha-1}\}\subset \calS_{a}^{\alpha}$ be an index set such that $\beta\notin \{i_1,\cdots, i_{\alpha-1}\}$ but $\beta\in \calS_{a}^{\alpha}$.  Note that when we pick the same set of blocks $\{B_{i_1,l_{i_1}},\cdots, B_{i_{\alpha-1},l_{i_{\alpha-1}}}\}$, it is impossible that the recovered subfiles $W^{E(\calP_{\beta})}_{\kappa_{\beta,l_{\beta}},\hat{L}_{\beta}[t_1]}$ and $W^{E(\calP_{\beta})}_{\kappa_{\beta,l_{\beta}},\hat{L}_{\beta}[t_2]}$ are the same since the points in $\hat{L}_\beta$ are distinct. Next, suppose that there exist sets of blocks $\{B_{i_1,l_{i_1}},\cdots, B_{i_{\alpha-1},l_{i_{\alpha-1}}}\}$ and $\{B_{i_1,l'_{i_1}}, \cdots, B_{i_{\alpha-1},l'_{i_{\alpha-1}}}\}$ such that $\{B_{i_1,l_{i_1}}, \cdots, B_{i_{\alpha-1},l_{i_{\alpha-1}}}\}\neq \{B_{i_1,l'_{i_1}}, \cdots, B_{i_{\alpha-1},l'_{i_{\alpha-1}}}\}$, but $\gamma\in \cap_{j=1}^{\alpha-1} B_{i_j,l_{i_j}}\setminus B_{\beta,l_{\beta}}$ and $\gamma\in \cap_{j=1}^{\alpha-1} B_{i_j,l'_{i_j}} \setminus B_{\beta,l'_{\beta}}$. This is a contradiction since this in turn implies that $\gamma \in \cap_{j=2}^{\alpha} B_{i_j,l_{i_j}}\bigcap \cap_{j=2}^{\alpha} B_{i_j,l'_{i_j}}$, which is impossible since two blocks from the same parallel class have an empty intersection.
	
	Finally we calculate the transmission rate. In Algorithm \ref{Alg:SignalAlphaCCP}, for each recovery set, we transmit $q^{k+1}-q^{k}$ equations and there are totally $\frac{zn}{\alpha}$ recovery sets. Since each equation has size equal to a subfile, the rate is given by
	\begin{align*}
	R&=(q^{k+1}-q^{k})\times \frac{zn}{\alpha}\times \frac{1}{zq^k}\\
	 &=\frac{n(q-1)}{\alpha}.
	\end{align*}

The $(n,k)$ linear block codes that satisfy the $(k,\alpha)$-CCP over $GF(q)$ correspond to a coded caching system with $K=nq$, $\frac{M}{N}=\frac{1}{q}$, $F_s=zq^k$ and have a rate $R=\frac{n(q-1)}{\alpha}$. Thus, the rate of this system is a little higher compared to the $(k,k+1)$-CCP system with almost the same subpacketization level.

However, by comparing Definitions \ref{def:MDSproperty} and \ref{def:kalphacc} it is evident that the rank constraints of the $(k,\alpha)$-CCP are weaker as compared to the $(k,k+1)$-CCP. Therefore, in general we can find more instances of generator matrices that satisfy the $(k,\alpha)$-CCP. For example, a large class of codes that satisfy the $(k,k)$-CCP are $(n,k)$ cyclic codes since any $k$ consecutive columns in their generator matrices are linearly independent \cite{lincostello}. Thus, $(n,k)$ cyclic codes always satisfy the $(k,k)$-CCP but satisfy $(k,k+1)$-CCP if they satisfy the additional constraints discussed in Claim \ref{claim:cyclic_k1}.

%

\subsection{Cyclic codes over $\mathbb{Z} \mod q$ \cite{blake1972codes}}
\label{sec:blake_codes}
%
%
%
First, we show that matrix $\bfT$ constructed by constructed by the approach outlined in Section \ref{sec:matrices_z_mod_q} still results in a resolvable design. Let $\Delta=[\Delta_0 \Delta_1 \cdots \Delta_{n-1}]$ be a codeword of the cyclic code over $\mathbb{Z}\mod q$, denoted $\calC$ where $q=q_1q_2\cdots q_d$, and $q_i, i = 1, \dots, d$ are prime. By using the Chinese remaindering map $\psi$ (discussed in Section \ref{sec:matrices_z_mod_q}), $\Delta$ can be uniquely mapped into $d$ codewords $\bfc^{(i)}, i = 1, \dots, d$ where each $\bfc^{(i)}$ is a codeword of $\calC^i$ (the cyclic code over $GF(q_i)$). Thus, the $b$-th component $\Delta_b$ can be mapped to $(\bfc^{(1)}_{b}, \bfc^{(2)}_{b}, \dots, \bfc^{(d)}_{b})$

Let $\bfG^i=[g_{ab}^{(i)}]$ represent the generator matrix of the code $\calC^i$. Based on prior arguments, it is evident that there are $q_i^{k_i-1}$ distinct solutions over $GF(q_i)$ to the equation $\sum_{a=0}^{k_i-1}\bfu_a g_{ab}^{(i)}= \bfc^{(i)}_{b}$. In turn, this implies that $\Delta_b$ appears $q_1^{k_1-1}q_2^{k_2-1}\cdots q_d^{k_d-1}$ times in the $b$-th row of $\bfT$ and the result follows.
%
%


Next we show any $\alpha$ blocks from distinct parallel classes of $\calP_{\calS_{a}^{k_{min}}}$  have $q_1^{k_1-\alpha}q_2^{k_2-\alpha}\cdots q_{d}^{k_d-\alpha}$ intersections, where $\alpha\le k_{min}$ and ${\calS^{k_{min}}_a}=\{(ak_{min})_n, (ak_{min}+1)_n,\cdots,(ak_{min}+k_{min}-1)_n\}$

Towards this end consider $B_{i_1, l_{i_1}}, \dots, B_{i_{\alpha}, l_{i_{\alpha}}}$ (where $i_j \in \calS^{k_{min}}_a, l_{i_j} \in \{0, \dots, q-1\}$) that are picked from $\alpha$ distinct parallel classes of  $\calP_{\calS^{k_{min}}_a}$. W.l.o.g. we assume that $i_1 < i_2 < \dots < i_{\alpha}$. Let $\calI =  \{i_1, \dots, i_{\alpha}\}$ and $\bfT_{\calI}$ denote the submatrix of $\bfT$ obtained by retaining the rows in $\calI$. We will show that the vector $[l_{i_1}~l_{i_2}~\dots~l_{i_{\alpha}}]^T$ is a column in $\bfT_{\calI}$ and appears $q_1^{k_1-\alpha}q_2^{k_2-\alpha}\cdots q_{d}^{k_d-\alpha}$ times.

Let $\psi_m(l_{i_j})$ for $m = 1, \dots, d$ represent the $m$-th component of the map $\psi$. Consider the $(n,k_1)$ cyclic code over $GF(q_1)$ and the system of equations in variables $\bfu_0, \dots, \bfu_{\alpha-1}$ that lie in $GF(q_1)$.
\begin{align*}
\sum_{b=0}^{\alpha-1}\bfu_{b}g_{bi_1}^{(1)}&= \psi_1(l_{i_1})-\sum_{b=\alpha}^{k_1-1}\bfu_{b}g_{bi_1}^{(1)},\\
\sum_{b=0}^{\alpha-1}\bfu_{b}g_{bi_2}^{(1)} &= \psi_1(l_{i_2})-\sum_{b=\alpha}^{k_1-1}\bfu_{b}g_{bi_2}^{(1)},\\
\mathrel{\makebox[\widthof{=}]{\vdots}}\\
\sum_{b=0}^{\alpha-1}\bfu_{b}g_{bi_{\alpha}}^{(1)} &= \psi_1(l_{i_{\alpha}})-\sum_{b=\alpha}^{k_1-1}\bfu_{b}g_{bi_{\alpha}}^{(1)}.
\end{align*}
By arguments identical to those made in Claim \ref{claim:MDSintersectionKAlpha} it can be seen that this system of equations has $q_1^{k_1-\alpha}$ solutions. Applying the same argument to the other cyclic codes we conclude that the vector $[l_{i_1},l_{i_2},\cdots, l_{i_{\alpha}}]$ appears $q_1^{k_1-\alpha}q_2^{k_2-\alpha}\cdots q_{d}^{k_d-\alpha}$ times in $\bfT_{\calI}$ and the result follows. 

\begin{IEEEbiographynophoto}{Li Tang} received his   B.E. degree in mechanical engineering from Beihang University, Beijing, China in 2011, and M.S. degree in electrical and information engineering from Beihang University, Beijing, China in 2014.  He is currently working towards the Ph.D degree in the Department of Electrical and Computer  Engineering at Iowa State University, Ames, IA, USA. His research interests include network coding and channel coding.
\end{IEEEbiographynophoto}

\begin{IEEEbiographynophoto}{Aditya Ramamoorthy}
(M'05) received the B.Tech. degree in electrical engineering from the Indian Institute of Technology, Delhi, in 1999, and the M.S. and Ph.D. degrees from the University of California, Los Angeles (UCLA), in 2002 and 2005, respectively. He was a systems engineer with Biomorphic VLSI Inc. until 2001. From 2005 to 2006, he was with the Data Storage Signal Processing Group of Marvell Semiconductor Inc. Since fall 2006, he has been with the Electrical and Computer Engineering Department at Iowa State University, Ames, IA 50011, USA. His research interests are in the areas of network information theory, channel coding and signal processing for bioinformatics and nanotechnology. Dr. Ramamoorthy served as an editor for the IEEE Transactions on Communications from 2011 -- 2015. He is currently serving as an associate editor for the IEEE Transactions on Information Theory. He is the recipient of the 2012 Early Career Engineering Faculty Research Award from Iowa State University, the 2012 NSF CAREER award, and the Harpole-Pentair professorship in 2009 and 2010.
\end{IEEEbiographynophoto}

 \end{document}